\documentclass[12pt, a4paper, reqno]{article}
\pdfoutput=1 % Needed for arXiv submission

\usepackage{mathpazo}

\pdfpagewidth=\paperwidth
\pdfpageheight=\paperheight
\usepackage{microtype}

\usepackage{setspace}
\usepackage{enumitem}
\setlist[itemize]{align=left,leftmargin=*,label=\raisebox{0.25ex}{\scriptsize$\bullet$},noitemsep}
\setlist[enumerate]{align=left,leftmargin=*,noitemsep}

\usepackage{graphicx}
\usepackage{epstopdf}
\usepackage{caption,subcaption,multirow,booktabs} % booktabs for \toprule, multirow to span rows in tables
\usepackage{tikz}
\usetikzlibrary{patterns}
\usetikzlibrary{shapes}

\usepackage{datetime} % Allows to typeset compile time with \currenttime
\newdateformat{monthyeardate}{\monthname[\THEMONTH] \THEYEAR} % Use \monthyeardate\today for April 2049
\newdateformat{datestyle}{{\THEDAY} \monthname[\THEMONTH] \THEYEAR} % Use \datestyle\today for 25 April 2049

\usepackage[titletoc,title]{appendix}

\usepackage{epigraph}

\usepackage[longnamesfirst]{natbib}
\usepackage[hidelinks]{hyperref} % [hidelinks] removes boxes from links

\hypersetup{
		colorlinks=true,
		linkcolor=black,
		citecolor=black,
		filecolor=black,
		urlcolor=black,
	}

\usepackage{amsmath,amssymb,amsfonts,amsthm,mathtools}

\usepackage[nameinlink]{cleveref}
\crefname{page}{p.}{pp.}
\crefname{equation}{equation}{equations}
\crefname{section}{section}{sections}
\crefname{subsection}{section}{sections}
\crefname{subsubsection}{section}{sections}
\crefname{appsec}{appendix}{appendices}
\crefname{supplsec}{supplemental appendix}{supplemental appendices}
\crefname{footnote}{footnote}{footnotes}
\crefname{figure}{figure}{figures}
\crefname{table}{table}{tables}
\crefname{theorem}{theorem}{theorems}
\crefname{proposition}{proposition}{propositions}
\crefname{lemma}{lemma}{lemmata}
\crefname{corollary}{corollary}{corollaries}
\crefname{remark}{remark}{remarks}
\crefname{observation}{observation}{observations}
\crefname{example}{example}{examples}
\crefname{fact}{fact}{facts}
\crefname{definition}{definition}{definitions}
\crefname{assumption}{assumption}{assumptions}
\crefname{exercise}{exercise}{exercises}
\crefname{notation}{notation}{notation}
\crefname{claim}{claim}{claims}
\crefname{conjecture}{conjecture}{conjectures}

\theoremstyle{plain}	% {definition} roman, {plain} italic
\newtheorem{theorem}{Theorem}
\newtheorem{proposition}{Proposition}
\newtheorem{lemma}{Lemma}
\newtheorem{corollary}{Corollary}
\newtheorem{claim}{Claim}

\theoremstyle{definition}	% {definition} roman, {plain} italic
\newtheorem{definition}{Definition}
\newtheorem{example}{Example}

\newtheorem{remark}{Remark}

\newtheoremstyle{named}
	{\topsep}									% ABOVESPACE
	{\topsep}									% BELOWSPACE
	{}												% BODYFONT
	{0pt}											% INDENT (empty value is the same as 0pt)
	{\bfseries}								% HEADFONT
	{}												% HEADPUNCT
	{5pt plus 1pt minus 1pt}	% HEADSPACE
	{\thmnote{#3}}						% CUSTOM-HEAD-SPEC
\theoremstyle{named}
\newtheorem{namedthm}{}

\DeclareMathOperator*{\argmax}{arg\,max}

\DeclareMathOperator*{\cav}{cav}
\DeclareMathOperator*{\supp}{supp}
\DeclareMathOperator*{\conv}{conv}

\newcommand\xqed[1]{%
  \leavevmode\unskip\penalty9999 \hbox{}\nobreak\hfill
  \quad\hbox{#1}}

\newcommand\demo{\xqed{$\diamond$}}

\begin{document}

\title{\sc Covert learning and disclosure\thanks{I am grateful to Piero Gottardi and Giacomo Calzolari for guidance and support, and to Nina Bobkova and Fr\'{e}d\'{e}ric Koessler for detailed comments and advice.
I thank, for useful comments and feedback, Doruk Cetemen, Wouter Dessein, Alfredo Di Tillio, Alessandro Ferrari, Daniel Garrett, Mathijs Janssen, Navin Kartik, Nenad Kos, David K. Levine, Andrea Mattozzi, Marco Ottaviani, Ludvig Sinander, Alessandro Spiganti, Emanuele Tarantino, Adrien Vigier and audiences at Bocconi, EUI, LUISS, Naples `Federico II'/CSEF, Rome `Tor Vergata', TSE, the 2019 Lisbon meetings in game theory and applications, the 2019 European winter meetings and 2020 world congress of the Econometric Society, the 2021 spring meeting of young economists and GAMES 2020. A previous version of this paper was titled `Communication with partially verifiable endogenous information'.}}

\author{Matteo Escud\'{e} \\ {LUISS}}

\date{25 October 2025}
\maketitle

\begin{abstract}
	I study a model of information acquisition and transmission
	in which the sender's ability to misreport her findings is limited.
	The sender learns covertly,
	so a key observation is that in equilibrium
	she must be deterred
	from undetectably \emph{worsening the meaning}
	of the messages she sends.
	This force substantially disciplines
	equilibrium beliefs and behavior:
	the receiver is maximally skeptical
	of what the sender claims
	and learns all discovered information.
	I exploit these equilibrium properties
	to characterize the sense in which the sender benefits
	from her claims being \emph{more verifiable},
	showing that this is akin to increasing her commitment power.
	Finally, I identify sender- and receiver-optimal falsification environments.
\end{abstract}

\section{Introduction}
\label{sec:intro}

To inform their choices,
individuals typically rely on the advice of experts
with access to superior information.
In the presence of conflicting interests,
the latter will attempt to influence the former's decisions in at least two ways:
by \emph{acquiring information selectively}
and---as far as is feasible---by \emph{misrepresenting their findings}.
This paper studies the interplay between these two modes of influence,
exploring how the expert's ability to misrepresent her discoveries
influences what information she acquires and, ultimately, transmits to the receiver.

These two persuasion channels are central to several important environments.
CEOs allocate resources within firms using information
gathered and reported by empire-building managers,
who favor investment in their own divisions
in spite of the overall profitability for the firm.
Prospective buyers base their purchasing decisions
on reports by market analysts and advisors;
in the presence of commissions (or other incentive schemes),
such experts will obtain and report information
in ways that maximize the chances of a sale.
Voters evaluate politicians by reading investigative reports and analyses
by partisan journalists,
who are selective in their coverage
and partial in their portrayal of events.

These settings share three key features,
which are also the essential components of the model I study.
First,
there is a substantial conflict of interest
between the `sender' and the `receiver'.
Because of this conflict,
the sender will use her ability
to control information production and transmission
to influence the receiver's decision in her favor,
making the latter---in turn---wary of the advice.

Second, the sender can gather information privately,
without substantial scrutiny by the receiver.
In applications,
this is the case because information production
is often formally delegated
(as in many manager-CEO relationships),
difficult to monitor
(as in the advisor-buyer case),
or secret because of social norms
(for example, journalistic sources are protected by anonymity).

Third, the sender faces constraints on the extent to which
she can misrepresent her findings.
This may be because some of her claims are \emph{verifiable}
by the receiver or by a trusted third party.
She might also have reputational concerns,
fearing that reports excessively different from what she finds
may be ultimately discovered,
damaging her reputation and future credibility.
Additionally, she might be bound by disclosure regulation
or contractual obligations.%
\footnote{For example, reputational concerns are the main (albeit imperfect) reason
behind the credibility of credit-rating agencies (see \citet{Mathis2009} for a discussion).
Disclosure regulation plays a crucial role in financial reporting (see \citet{Leuz2016} for a survey).}

The model I develop incorporates these three ingredients
by marrying the verifiable-disclosure literature's
flexible approach to constrained communication
with Bayesian persuasion's belief-based methods.
My set-up can therefore be viewed from two different angles.
It is a sender-receiver game with (partially) verifiable information,
in which the sender's type distribution is determined
by her own equilibrium covert information acquisition,
rather than fixed and common knowledge.
It is also a model of Bayesian persuasion with limited commitment,
in which the sender can acquire information freely and covertly,
but can only commit to sending a message from
a fixed set for each realization of her private signal.

Combining these two approaches proves fruitful:
I obtain an unraveling result
across all possible misreporting constraints in a large class,
which highlights how these constraints
shape \emph{what} information the sender chooses to acquire.
More specifically, since the sender learns covertly,
she must be deterred in equilibrium from
undetectably \emph{worsening the meaning}
of the messages she is expected to send.
This serves as a key disciplining force,
ensuring that the sender seeks only information
that she can then credibly transmit,
given the misreporting constraints.

This insight, in turn, allows me to address
comparative statics questions,
delineating the connection between
\emph{more verifiability}
of the sender's reporting
and \emph{more commitment power}
in Bayesian persuasion problems
with limited sender commitment.
It also delivers a characterization of optimal verifiability environments:
for the sender, this is a characterization
of the constraints that allow her to obtain
the Bayesian persuasion outcome;
for the receiver, they describe constraints that incentivize
the sender to acquire (and transmit) full information in equilibrium.

\subsection{Example and overview of the results}
\label{sec:intro:example-overview}
Two players, a sender and a receiver,
are initially uninformed about a binary state,
taking values in $\{0,1\}$.
They share a common prior belief $\overline{p} = 1/4$ that the state is $1$.
The sender can covertly and freely acquire information about the state
by privately choosing
a signal---that is: a mean-$1/4$ distribution of posterior beliefs---and
observing its realization.

An important primitive of the model is the \emph{verifiability structure},
which describes what the sender can say to the receiver
following any given outcome of her information acquisition.
Formally, if the sender discovers that the state is $1$ with probability $s \in [0,1]$,
she can send a message from the set $M(s)$ to the receiver.%
\footnote{The fact that the sender's posterior belief
uniquely determines what she can communicate
is an assumption of the model,
discussed in detail in \Cref{sec:model:bdm}.}

For the purposes of this example, suppose that
there are only two possible messages that the sender can use:
call them $m_\emptyset$ and $m_G$.
Assume that $m_\emptyset \in M(s)$ for all $s \in [0,1]$,
so that it can be interpreted as a `silent' message,
meaning that the sender can always
choose to say nothing to the receiver.
Assume instead that $m_G \in M(s)$ iff $s \in [0.5,1]$,
meaning that $m_G$ proves to the receiver
that the state has at least a 50\% chance of being $1$.

Assume that the sender is biased in favor of state $1$
in the sense that she (weakly) benefits from inducing
higher beliefs about the state in the receiver.
Specifically, denote by $v(p)$ the payoff
the sender obtains from inducing belief $p \in [0,1]$ in the receiver
and suppose that $v(p) = 0$ for $p < 0.4$,
$v(p) = 1$ for $p \in [0.4,0.8)$
and $v(p) = 3$ for $p \ge 0.8$.

The payoff $v$ is a reduced-form representation
of the downstream interaction between the sender and the receiver.%
\footnote{In this example, a standard microfoundation
involves the receiver facing a decision problem under uncertainty
with three actions: $a_L$, $a_M$ and $a_H$.
Her payoff is such that she optimally chooses
$a_L$ at beliefs strictly below $0.4$,
$a_M$ at beliefs in $[0.4,0.8)$
and $a_H$ at beliefs at or above $0.8$.
The sender has state-independent preferences over the receiver's actions,
obtaining payoffs of $0$, $1$ and $3$ from actions $a_L$, $a_M$ and $a_H$, respectively.}
The latter's behavior is not modeled explicitly:
the receiver is an active player only insofar as we will require
her posterior belief after every message to be
always consistent with the verifiability structure
(so, after observing $m_G$,
she must believe the state
has at least a $50\%$ chance of being equal to $1$)
and formed using Bayes' rule on the equilibrium path.

\paragraph{Equilibrium unraveling and skepticism.}
Suppose first that the sender chooses to acquire information about the state
through a signal `splitting' her prior into posteriors $0$ and $0.8$.
Note that this would be the equilibrium signal
if the sender's information acquisition
were public---as in \citet{Kamenica2011}---attaining
an expected payoff of $(\cav v) (1/4)$,
as depicted in \Cref{fig:3-act-example:bp}.

\begin{figure}[ht]
	\centering
	\begin{subfigure}[h]{0.45\textwidth}
	\centering
		\begin{tikzpicture}[scale=0.5]

  % === Axes ===
  \draw (0,0) -- (10,0);       % horizontal axis
  \draw[->] (0,0) -- (0,5.5);   % vertical axis

  % === Tick marks & labels ===
  \draw (0,0.15)  -- (0,-0.15)   node[below] {\small $0$};
  \draw (10,0.15) -- (10,-0.15)  node[below] {\small $1$};
  \draw (4,0.15)  -- (4,-0.15)   node[below] {\small $0.4$};
  \draw (2.5,0.15)-- (2.5,-0.15) node[below] {\small $\overline{p}$};
  \draw (8,0.15)  -- (8,-0.15)   node[below] {\small $0.8$};
  \draw (0.15,10/6) -- (-0.15,10/6) node[left] {\small $1$};
  \draw (0.15,10/2) -- (-0.15,10/2) node[left] {\small $3$};

  % === Value function (light gray) ===
  \draw[very thick, gray] (0,0) -- (4,0);
  \draw[very thick, gray] (4,10/6) -- (8,10/6);
  \draw[very thick, gray] (8,10/2) -- (10,10/2);
  % Markers (light gray)
  \draw[fill, gray] (4,10/6) circle (0.075);
  \draw[fill, gray] (8,10/2) circle (0.075);
  \draw[draw=gray, fill=white] (4,0)     circle (0.075);
  \draw[draw=gray, fill=white] (8,10/6)  circle (0.075);

  % === Cav (dotted) ===
  \draw[ultra thick, dotted] (0,0) -- (8,10/2) -- (10,10/2);
  \draw[draw=black, fill=black] (2.5,1.5625) circle (0.15); % Equilibrium payoff

  % === Messages (hidden with phantom) ===
  \phantom{
    \draw (0,-1.5) -- (10,-1.5) node[right] {\small $\bar{m}$};
    \draw (5,-2.2) -- (10,-2.2) node[right] {\small $m_G$};
    \draw (0,-1.4) -- (0,-1.6);
    \draw (10,-1.4) -- (10,-1.6);
    \draw (5,-2.1) -- (5,-2.3);
    \draw (10,-2.1) -- (10,-2.3);
  }

\end{tikzpicture}
		\caption{\footnotesize Bayesian persuasion.}
		\label{fig:3-act-example:bp}
	\end{subfigure}
	\quad
	\begin{subfigure}[h]{0.45\textwidth}
	\centering
		\begin{tikzpicture}[scale=0.5]

  % === Axes ===
  \draw (0,0) -- (10,0);        % horizontal axis
  \draw[->] (0,0) -- (0,5.5);   % vertical axis

  % === Tick marks & labels ===
  \draw (0,0.15) -- (0,-0.15)   node[below] {\small $0$};
  \draw (10,0.15) -- (10,-0.15) node[below] {\small $1$};
  \draw (2.5,0.15) -- (2.5,-0.15) node[below] {\small $\overline{p}$};
  \draw (5,0.15) -- (5,-0.15)   node[below] {\small $0.5$};
  \draw (0.15,10/6) -- (-0.15,10/6) node[left] {\small $1$};
  \draw (0.15,10/2) -- (-0.15,10/2) node[left] {\small $3$};

  % === Value function (light gray) ===
  \draw[very thick, gray] (0,0) -- (4,0);
  \draw[very thick, gray] (4,10/6) -- (8,10/6);
  \draw[very thick, gray] (8,10/2) -- (10,10/2);
  % Markers (light gray)
  \draw[fill, gray] (4,10/6) circle (0.075);
  \draw[fill, gray] (8,10/2) circle (0.075);
  \draw[draw=gray, fill=white] (8,10/6) circle (0.075);
  \draw[draw=gray, fill=white] (4,0)     circle (0.075);

  % === Skeptical value (black) ===
  \draw[thick] (0,0) -- (5,0);
  \draw[thick] (5,10/6) -- (10,10/6);
  % Markers (black)
  \draw[draw=black, fill=white] (5,0)     circle (0.075);
  \draw[fill] (5,10/6) circle (0.075);

  % === Cav (dotted) ===
  \draw[ultra thick, dotted] (0,0) -- (5,10/6) -- (10,10/6);
  \draw[draw=black, fill=black] (2.5,10/12) circle (0.15); % Equilibrium payoff

  % === Messages ===
  \draw (0,-1.5) -- (10,-1.5) node[right] {\small $m_\emptyset$};
  \draw (5,-2.2) -- (10,-2.2) node[right] {\small $m_G$};
  % Message ticks
  \draw (0,-1.4) -- (0,-1.6);
  \draw (10,-1.4) -- (10,-1.6);
  \draw (5,-2.1) -- (5,-2.3);
  \draw (10,-2.1) -- (10,-2.3);

\end{tikzpicture}
		\caption{\footnotesize Covert learning and disclosure.}
		\label{fig:3-act-example:cld}
	\end{subfigure}
	\caption
		{\small On the left: the sender payoff function $v$
		(solid gray) and
		its concave envelope
		(thick, dotted)
		as a function of the receiver's belief.
		On the right:
		verifiability structure $M$
		(below the horizontal axis),
		the sender's skepticism-adjusted value function
		(solid black)---overlaid over $v$
		(solid gray)---and its concave envelope
		(thick, dotted).
		The large black dot corresponds to
		the sender's equilibrium expected payoff in each case.}
	\label{fig:3-act-example}
\end{figure}

This cannot be an equilibrium signal in my model.
To see why, suppose indeed the sender chose such a signal in equilibrium,
sending message $m_\emptyset$ when her type is $0$
(i.e., when she privately discovers that the state is $0$ for sure)
and message $m_G$ when her type is $0.8$
(i.e., when the signal realization is $s = 0.8$).
In this candidate equilibrium,
the sender would obtain a payoff of $0$ following message $m_\emptyset$,
and a payoff of $3$ following message $m_G$.

However, the sender could profitably deviate
to a signal splitting the prior
between $0$ and $0.8 - \delta$ (for $\delta < 0.3$);
in the former case sending message $m_\emptyset$
and in the latter sending message $m_G$.
This deviation is not detectable by the receiver
(since she does not observe the sender's choice of signal)
and it increases the probability of sending $m_G$
while decreasing the probability of sending $m_\emptyset$,
thus strictly increasing the sender's expected payoff.
This kind of deviation---which \emph{worsens the meaning}
of message $m_G$---plays a significant role
in shaping equilibrium behavior.

Indeed, only a signal splitting the prior
into posteriors $0$ and $0.5$
is not prone to such a deviation.
These are the lowest sender types consistent
with messages $m_\emptyset$ and $m_G$, respectively,
and therefore these messages' meaning
cannot be worsened any further.
This signal is part of what I call
an `unraveling' equilibrium,
in which sender type $0$ sends message $m_\emptyset$,
and sender type $0.5$ sends $m_G$.
Anticipating this behavior by the sender,
the receiver is maximally skeptical,
assigning to each message
the lowest sender belief
that is consistent with it:
belief $0$ following message $m_\emptyset$ and
belief $0.5$ following message $m_G$.
These correspond precisely to the sender's beliefs,
meaning that all acquired information is revealed to the receiver.
This equilibrium is depicted in \Cref{fig:3-act-example:cld}.

The three features of the unraveling equilibrium just described
(on-path sender types being `lowest-consistent'
with the messages they send,
the receiver's maximal skepticism and
full revelation of acquired information)
generalize beyond the specific payoff function $v$
and verifiability structure $M$ of the example.
The first result of the paper
(\Cref{thm:unraveling})
shows that, when $v$ is non-decreasing,
all equilibria involving some variation in the receiver's belief
which is payoff-relevant for the sender
are `unraveling' in the sense that they satisfy these three properties.

A direct consequence is that
the sender's equilibrium expected payoff
is the same across all unraveling equilibria and
can be obtained as the concave envelope
of a skepticism-adjusted value of posteriors.
This function is depicted in \Cref{fig:3-act-example:cld}
in the context of the example,
and this observation is generalized in \Cref{cor:value}.

\paragraph{Uniqueness of unraveling.}
In the example, the unraveling equilibrium just presented
is, in fact, the only equilibrium.
\Cref{thm:pnbp} provides a mild sufficient condition
on the primitives of the environment
ensuring that all equilibria are unraveling
and that therefore the sender's equilibrium expected payoff is unique.

The condition is that the sender must be able to
`prove news better than the prior' (PNBP).
This property holds if there exists a message that,
even when interpreted most skeptically by the receiver,
leads to a payoff for the sender strictly above what she would obtain at the prior.
In the example, PNBP holds because message $m_G$---even if followed
by a receiver belief of $0.5$---leads to a payoff of $1$ for the sender,
which is strictly higher than 
her payoff of $0$ at the prior.

At a high level, PNBP ensures that all equilibria are unraveling
because a message guaranteeing the sender
a payoff strictly above the prior
always opens the door to a deviation 
that worsens the meaning of some message,
unless the receiver is maximally skeptical on-path.

\Cref{thm:pnbp} also shows that, in contrast,
if the sender cannot prove news better than the prior,
there is always a sender-preferred equilibrium
in which she acquires and transmits no information.
This demonstrates how her bias
makes information control worthless,
absent some minimal amount of verifiability,
as captured by PNBP.

\paragraph{More verifiability and more commitment power.}
A natural question to ask is whether (and in which sense)
the sender benefits from being able to `prove more'.
For concreteness, suppose that in the example
the sender has access to a new message, call it $m_E$,
which she can only send if she discovers the \emph{excellent} news that
the state has at least a 90\% chance of being high
(i.e., $m_E \in M(s)$ iff $s \in [0.9,1]$).

It can be easily verified that in equilibrium,
the sender will now choose to split the prior between $0$ and $0.9$,
sending messages $m_\emptyset$ and $m_E$, respectively.
This leads to an increase in the sender's equilibrium expected payoff
compared to the case in which only $m_\emptyset$ and $m_G$ were available,
as depicted in \Cref{fig:3-act-example-more-ver}.

\begin{figure}[ht]
	\centering
	\begin{tikzpicture}[scale=0.5]

  % === Axes ===
  \draw (0,0) -- (10,0);       % horizontal axis
  \draw[->] (0,0) -- (0,5.5);   % vertical axis

  % === Tick marks & labels ===
  \draw (0,0.15)  -- (0,-0.15)   node[below] {\small $0$};
  \draw (10,0.15) -- (10,-0.15)  node[below] {\small $1$};
  \draw (5,0.15)  -- (5,-0.15)   node[below] {\small $0.5$};
  \draw (2.5,0.15)-- (2.5,-0.15) node[below] {\small $\overline{p}$};
  \draw (9,0.15)  -- (9,-0.15)   node[below] {\small $0.9$};
  \draw (0.15,10/6) -- (-0.15,10/6) node[left] {\small $1$};
  \draw (0.15,10/2) -- (-0.15,10/2) node[left] {\small $3$};

  % === Value function (light gray) ===
  \draw[very thick, gray] (0,0) -- (4,0);
  \draw[very thick, gray] (4,10/6) -- (8,10/6);
  \draw[very thick, gray] (8,10/2) -- (10,10/2);
  % Markers (light gray)
  \draw[fill, gray] (4,10/6) circle (0.075);
  \draw[fill, gray] (8,10/2) circle (0.075);
  \draw[draw=gray, fill=white] (4,0)     circle (0.075);
  \draw[draw=gray, fill=white] (8,10/6)  circle (0.075);

  % === Skeptical value (black) ===
  \draw[thick] (0,0) -- (5,0);
  \draw[thick] (5,10/6) -- (9,10/6);
  \draw[thick] (9,10/2) -- (10,10/2);
  % Markers (black)
  \draw[draw=black, fill=white] (5,0)     circle (0.075);
  \draw[fill] (5,10/6) circle (0.075);
  \draw[draw=black, fill=white] (9,10/6)  circle (0.075);
  \draw[fill] (9,10/2) circle (0.075);

  % === Cav (dotted) ===
  \draw[ultra thick, dotted] (0,0) -- (9,10/2) -- (10,10/2);
  \draw[draw=black, fill=black] (2.5,1.3889) circle (0.15); % Equilibrium payoff

  % === Messages ===
  \draw (0,-1.5) -- (10,-1.5) node[right] {\small $m_\emptyset$};
  \draw (5,-2.2) -- (10,-2.2) node[right] {\small $m_G$};
  \draw (9,-2.9) -- (10,-2.9) node[right] {\small $m_E$};
  % Message ticks
  \draw (0,-1.4) -- (0,-1.6);
  \draw (10,-1.4) -- (10,-1.6);
  \draw (5,-2.1) -- (5,-2.3);
  \draw (10,-2.1) -- (10,-2.3);
  \draw (9,-2.8) -- (9,-3.0);
  \draw (10,-2.8) -- (10,-3.0);

\end{tikzpicture}
	\caption{\small Adding message $m_E$ makes the sender better off.}
	\label{fig:3-act-example-more-ver}
\end{figure}

What is the appropriate notion of `more verifiability'
which ensures that the sender's payoff
improves in equilibrium more generally?
\Cref{thm:cs} provides a pre-order
on the set of verifiability structures that characterizes
those shifts that are beneficial for the sender
when she can prove news better than the prior.
The necessary and sufficient condition is easy to describe:
the set of sender types that are able to separate
from all lower types with a single message,
regardless of the receiver's belief, must grow in the sense of set inclusion.
This guarantees that any equilibrium distribution of receiver beliefs
inducible in the lower-ranked verifiability structure
is also inducible in the higher-ranked one.

If the sender cannot prove news better than the prior---meaning that
unraveling equilibria need not exist---I show instead
that shifts towards
more verifiability may in fact hurt the sender ex ante
(\Cref{sec:cs:more-ver-hurt}).
This is caused by the lack of observability
of the sender's information acquisition strategy.
In a more constrained environment,
the equilibrium meaning of some messages
may worsen with more verifiability:
if the sender is known to be able to prove good news,
any claim of ignorance she makes
will no longer be taken at face value by the receiver.

The comparative statics result is useful for two reasons.
First, the pre-order makes it possible to compare
a rich variety of common verifiability environments.
For example, according to the pre-order,
cheap talk is minimal,
while full-verifiability structures are maximal.
Interval-partition verifiability structures
(\Cref{ex:partition})---capturing settings
with `grades' or `ratings'
that certifiably pool together sets of types---are ranked
if and only if one is a refinement of the other
(\Cref{ex:partition-refine} and \Cref{re:cs-examples}).
Adding a new message as in the example---thereby allowing the sender
to credibly communicate
her acquired information more finely---always leads to
a higher-ranked structure
(\Cref{ex:add-message} and \Cref{re:cs-examples}).

Second, the result can be viewed
as providing a tool for comparing Bayesian persuasion
environments with different forms of sender commitment.
In this context, it characterizes
the appropriate notion of `increasing commitment power'
which is valuable for the sender.

\paragraph{Optimal verifiability.}
The comparative statics result also leads directly to a characterization
of sender-optimal verifiability structures (\Cref{prop:sender-optimal}),
which allow the sender to attain the full-commitment value in equilibrium.
The result therefore highlights how
the full commitment assumption in the canonical Bayesian persuasion setting
can be relaxed, without affecting the sender's equilibrium payoff.
The necessary and sufficient condition states
that all sender types must be able to separate from all lower ones
by means of a single message,
no matter what the receiver believes.
Given the monotone bias of the sender,
this guarantees that separation is possible for all types in equilibrium,
allowing the sender to replicate the Bayesian persuasion outcome.

Finally, I characterize structures
that lead to full information provision in equilibrium
(\Cref{prop:full-info}).
These are receiver-optimal,
provided that the receiver
values information about the state.
They take a simple form:
only the very best news
(that is: claiming that the state is equal to $1$ with certainty)
is considered credible, while all inconclusive news is fully unverifiable.
Every structure of this kind incentivizes
perfect information acquisition
by the sender,
which she then fully reveals in equilibrium.

\subsection{Literature}
\label{sec:intro:lit}

There is an extensive literature studying
the role of verifiability (or lack thereof)
in sender-receiver games.
A number of papers have pointed out that the classical unraveling result
(\citet{Viscusi1978}, \citet{Grossman1980}, \citet{Grossman1981} and \citet{Milgrom1981})
breaks down when the sender's private information is not fully verifiable,
in the sense that some sender types do not have access to messages that
allow them to sustain separation in equilibrium.%
\footnote{This approach traces back to
\citet{Dye1985} and \citet{Jung1988} in sender-receiver games
and to \citet{Green1986} in a mechanism design setting.
Important contributions include
\citet{Okuno-Fujiwara1990},
\citet{Lipman1995}
and, more recently,
\citet{Mathis2008}, \citet{Hagenbach2014} and \citet{Rappoport2024}.
Complete lack of verifiability is studied in the literature on cheap talk,
initiated by \citet{Crawford1982}.}
In this literature, the sender's type distribution is fixed and common knowledge;
in my model, the players start equally uninformed
and the sender can privately acquire information
before communicating with the receiver.
This paper therefore explores the role
of (partially) verifiable information
in a context in which what the sender privately knows
is a covert strategic decision on her part.

In doing so, this paper contributes to a growing literature
which augments sender-receiver games
by endogenizing the sender's private information.
This includes papers studying sender learning in cheap talk models,%
\footnote{See, for example, \citet{Pei2015} and \citet{Deimen2019}.
Earlier, \citet{Austen-Smith1994} explored this question
in a setting in which the sender
either learns the value of the state perfectly, or not at all;
in his setting, information acquisition is covert,
but the sender can prove that she is informed.}
settings with full verifiability%
\footnote{An early example is \citet{Matthews1985},
in which the sender faces a binary choice
between a fully informative and verifiable signal
and a fully uninformative and unverifiable signal.
More recently, \citet{Gentzkow2017} show that
under full verifiability
the sender perfectly reveals all acquired information.}
and gradual evidence acquisition.%
\footnote{For example,
\citet{Argenziano2016}, \citet{Felgenhauer2017},
\citet{Janssen2018}, \citet{Herresthal2022}.}
This paper studies the case in which
the sender's covertly acquired private information
is only partially verifiable.
For this reason, it is closest to a pair of papers
studying covert information acquisition
in \citet{Dye1985}-like environments:
\citet{Ben-Porath2018} and, in particular, \citet{DeMarzo2019}.%
\footnote{Information acquisition with Dye evidence
is also studied in \citet{Kartik2017} and \citet{Shishkin2024}.}

As in these papers, my model studies
how the lack of observability of the sender's information acquisition
determines what she chooses to learn and communicate with the receiver.
There are two substantial differences in our approaches.

The first concerns the class of verifiability structures studied.
In the Dye setting these papers consider,
the only unverifiable message is `non-disclosure'.
Its meaning must therefore be determined in equilibrium
and will be shaped by the sender's ability to conceal negative evidence
and to covertly change the way she decides to acquire information.
In my model, the meaning of several messages---potentially
of all of them---must be
jointly determined in equilibrium.%
\footnote{More formally, the receiver's belief following any $m$
for which $M^{-1}(m)$ is not a singleton is determined in equilibrium.
With Dye evidence, only the `non-disclosure' message, call it $m_\emptyset$,
is such that $M^{-1}(m_\emptyset)$ contains more than one element.}
My approach to partial verifiability allows
me to capture Dye-style evidence (see \Cref{ex:dye})
but also to speak directly to a variety of canonical settings
in the communication literature
which Dye evidence cannot capture.
For example, my formalization of verifiability
encompasses settings ranging from
mandatory disclosure (and, more
broadly, full verifiability) to
cheap talk (see \Cref{ex:mandatory-disclosure,ex:full-verifiability,ex:cheap-talk,ex:partition}).

The second key difference is that in these papers
the sender's payoff is affine in the receiver's belief.
It follows that the `selective information acquisition' persuasion motive
present in my model
(and central to the Bayesian persuasion literature)
is absent from these papers.

One important dimension in which my setting is less general
than theirs is that they study problems
in which the sender is constrained
in the signals she can choose.
In my model the sender is fully unrestricted
in her information acquisition strategy.
I discuss this in further detail in the context of \Cref{ex:dye} on \cpageref{ex:dye-cont}.

This paper also contributes to the literature exploring the consequences of
relaxing the commitment assumption in \citet{Kamenica2011}.
Some papers allow the sender to alter negative results with some positive probability
(for example, \citet{Lipnowski2022}, \citet{Frechette2022} and \citet{Min2021})
while others study models of Bayesian persuasion with lying costs.
In this latter category,
\citet{Nguyen2021} analyze a setting in which the sender's
information acquisition strategy is observable by the receiver
while in \citet{Guo2021}'s set-up, as in mine, it is not.
Both provide conditions on the costs
such that the sender-preferred equilibrium exhibits full information revelation.

In contrast to these papers, in my model
the sender faces a constraint on what she can or cannot communicate
depending on what she discovers.
This distinct formalization of partial commitment
allows me to bridge the gap
between the Bayesian persuasion setting
and models of verifiable disclosure.
This, in turn, leads to a characterization of when
`more verifiability' of the sender's acquired information 
is akin to her having `more commitment power'.

Finally, this paper is related
to the three applied literatures mentioned in the introduction:
information control within organizations,%
\footnote{\citet{GibbonsMatouschekRoberts2013} survey the literature on conflicting interests within organizations and their effect on the production and transmission of information.}
acquisition and disclosure of product information,%
\footnote{\citet{Dranove2010} survey the product-quality disclosure literature. In the context of acquisition and disclosure of financial information, \citet{Lin1998} and \citet{Michaely1999} provide empirical evidence of how financial analysts are biased in their recommendations when their employers have an underwriting relation with the firm under review. The financial reporting literature is surveyed in \citet{Leuz2016}.}
and media bias.%
\footnote{A large and active empirical literature documents media bias and its consequences for news production, reporting and electoral behavior. For example, \citet{Groseclose2005} highlight the \emph{presence} of media bias and \citet{DellaVigna2007} emphasize its influential role. Evidence of partial and selective coverage of events is also widely documented, for example in \citet{Puglisi2011} and \citet{Larcinese2011}. See \citet{Gentzkow2015} for a survey.}

\paragraph{Roadmap.} I introduce the model in the next section.
Unpersuasive and unraveling equilibria are studied in \Cref{sec:equilibria},
along with a discussion of the value of persuasion
in the context of my model of disclosure with covert learning.
\Cref{sec:cs} contains the comparative statics result,
characterizing changes in verifiability that are beneficial for the sender,
and the characterization of sender-optimal verifiability structures.
Receiver-optimal structures are characterized in \Cref{sec:receiver-optimal}.

\section{Model}
\label{sec:model}

There is an unknown binary state of the world $\omega$, taking values in $\{0,1\}$.
A sender and a receiver initially assign prior probability $\overline{p} \in [0,1]$ to $\omega = 1$.
The sender can acquire information about the state by choosing a signal and observing its realization.
Information acquisition is covert in the sense that
the receiver observes neither the sender's chosen signal, nor its realization.
Formally, a feasible signal is a pair $(\mathcal{S},\pi)$
where $\mathcal{S}$ is a finite set of possible realizations
and $\pi \in \Delta (\{0,1\} \times \mathcal{S})$
is a joint distribution
such that the marginal distribution of the state
coincides with the prior $\overline{p}$:
$\sum_{s \in \mathcal{S}} \pi(1,s) = \overline{p}$.%
\footnote{We will soon restrict $\mathcal{S}$ to be a subset of $[0,1]$, see \Cref{re:unbias-belief-mes}.}

After acquiring information, the sender communicates with the receiver.
Her acquired private information determines what she is able to say:
if the sender's chosen signal is $(\mathcal{S},\pi)$ and its realization is $s \in \mathcal{S}$,
she must send a message from the non-empty finite set $M(\mathcal{S},\pi;s)$ to the receiver.
I refer to the mapping $M$, assigning sets $(M(\mathcal{S},\pi;s))_{s \in \mathcal{S}}$ to every signal $(\mathcal{S},\pi)$,
as the \emph{verifiability structure}, which is common knowledge among the two players.
$M$ describes how the sender's acquired private information constrains how she can communicate with the receiver.
In doing so, it also captures what the sender can prove about what she learned to the receiver.

Having observed the sender's message, the receiver updates her belief about the state.
The sender obtains a payoff $v(p)$ if the receiver holds a posterior belief $p \in [0,1]$.
The function $v: [0,1] \to \mathbb{R}$ is assumed to be non-decreasing,
capturing that the sender prefers inducing higher beliefs in the receiver,
and upper semi-continuous.
The sender maximizes the expected value of $v$.

The preferences of the receiver are not modeled explicitly as they play no role in the analysis.
The receiver can be interpreted either as a literal player,
who takes a payoff-relevant action following the sender's message,
or as a `market' in which the sender obtains return $v(p)$ at belief $p$.

\subsection{Beliefs determine messages}
\label{sec:model:bdm}

Operationalizing the notion of verifiability structure requires
articulating how the sender's private information shapes what she can say
and, vice versa, how the receiver's beliefs are restricted
following a given message from the sender.
I do so by restricting attention to verifiability structures
in which \emph{beliefs determine messages}, defined as follows.

\begin{definition}[Beliefs determine messages]
	\label{def:bdm}
	$M$ is such that \emph{beliefs determine messages}
	iff for every pair of signals
	$(\mathcal{S}, \pi)$ and $(\mathcal{S}', \pi')$,
	and every $s \in \mathcal{S}$, $s' \in \mathcal{S}'$,
	if $\pi(\cdot|s) = \pi'(\cdot|s')$
	then $M(\mathcal{S},\pi;s) = M(\mathcal{S}',\pi';s')$.
\end{definition}

If beliefs determine messages,
it is only the information \emph{about the state}
contained in the signal realization---that is:
the sender's own private posterior belief---that determines what messages the sender can use.

The underlying assumption is that the sender cannot \emph{directly} prove anything
about how she chose to acquire information
(that is: about the signal $(\mathcal{S},\pi)$),
but only about what she learned regarding the state.
Observe that proving information about the state
\emph{indirectly} proves something about the sender's choice of signal.
That is, if a certain message $m$ proves to the receiver that
the sender's belief lies in some set $S_m \subseteq [0,1]$,
the sender is also proving that she acquired information
in a way that assigns positive probability
to her obtaining a belief in $S_m$.
The beliefs determine messages assumption
imposes that such a message $m$ proves
nothing more than this to the receiver.

This assumption formalizes a feature of many applications:
the details of the information-gathering process are typically difficult to verify
while information about the state is---at least in part---verifiable.
The reason why it is hard to provide credible information
about the learning process may be technological:
for example, the information-gathering procedure may be long and difficult to monitor.
It may also be due to social norms or legal constraints:
for example, journalistic sources are typically protected by anonymity.

In contrast, information about the state can more easily be verified,
for example by a trusted third party
(such as an auditor, a fact-checker or an accreditation body) 
who certifies what the sender reports about the payoff-relevant state.
Importantly, such certification typically does not coincide
with full disclosure of what the sender discovers:
it may involve coarse grades or thresholds,
and perhaps leave the sender the possibility of redacting a report if it reveals bad news.
\Cref{sec:model:examples} below provides more detailed examples of such verifiability structures.

Importantly, abstracting away
from this secondary layer of communication has a methodological payoff,
as it allows the problem to be studied using a belief-based approach.

\begin{remark}
	\label{re:unbias-belief-mes}
	Given the restriction to verifiability structures
	in which beliefs determine messages,
	it follows that:
	\begin{enumerate}[label=(\roman*)]
		\item \label{re:unbias-belief-mes:unbias}
		It is without further loss to restrict the sender to choosing signals
		taking values in $[0,1]$ and such that each realization
		coincides with the sender's own private posterior belief about the state
		(i.e., such that $\pi(1|s) = s$ when $s \in [0,1]$ realizes).
		I shall frequently refer to $s$ as the sender's \emph{type}.
		\item \label{re:unbias-belief-mes:belief-mes}
		We can lighten notation further by omitting
		the dependence of $M$ on the signal $(\mathcal{S},\pi)$.
		Henceforth $M$ will simply denote a mapping assigning to each $s \in [0,1]$
		a non-empty finite set of messages $M(s)$.
	\end{enumerate}
\end{remark}

Before turning to examples,
I state a closure assumption
on the verifiability structures considered,
which is maintained throughout the analysis.
For a given verifiability structure $M$,
let $\mathcal{M} \equiv \cup_{s \in [0,1]} M(s)$
denote the set of all possible messages.
For any given message $m \in \mathcal{M}$,
let
\begin{equation*}
	M^{-1}(m) \equiv \{ s \in [0,1] : m \in M(s) \}
\end{equation*}
denote the set of sender types
that can send message $m$.
We assume that this set always contains its infimum:
\begin{equation}
	\label{eq:inf-M}
	\inf M^{-1}(m) \in M^{-1}(m) \quad \text{for every} \quad m \in \mathcal{M}.
\end{equation}
Therefore, for any given message,
there is a well-defined lowest type
that could have sent it.
This condition ensures that
the receiver's `maximally skeptical' belief---which
will be used throughout the analysis and is based
on this lowest type---is always well-defined.

\subsection{Examples of verifiability structures}
\label{sec:model:examples}

The following examples illustrate some notable cases
of verifiability structures in which beliefs determine messages.

\begin{example}[Mandatory disclosure]
	\label{ex:mandatory-disclosure}
	Suppose the sender is forced to reveal
	whatever she has discovered about the state.
	This could be due to a deliberate (ex ante) choice by the sender
	(e.g., by hiring a trusted third-party
	to certify that no acquired information is withheld)
	or to stringent disclosure regulation.
	Formally, we say $M$ is the \emph{mandatory disclosure}
	verifiability structure if $M(s) = \{s\}$ for each $s \in [0,1]$,
	capturing that the sender
	is forced to truthfully report her belief about the state.
	It is immediate (for more details, see the discussion in \Cref{sec:cs})
	that this verifiability structure
	allows the sender to attain \citet{Kamenica2011}'s `full-commitment' value in equilibrium.
	\demo
\end{example}

\begin{example}[Full verifiability]
	\label{ex:full-verifiability}
	Mandatory disclosure as described in the previous example
	is a special case of a verifiability structure in which
	the sender is able to prove precisely what she discovered to the receiver.
	More generally, the sender might have the \emph{option} of doing so,
	but may also be able to misreport or conceal her findings.
	For example, she might hire a trusted third-party certifier,
	but retain the option of preventing him from making any disclosure.%
	\footnote{We can describe this particular situation with the verifiability structure
	$M(s) = \{s, m_\emptyset\}$ for each $s \in [0,1]$.
	Therefore, upon learning $s$,
	the sender can either allow
	the third-party certifier to make a (truthful) disclosure of $s$
	or can prevent him from reporting anything,
	sending the null message $m_\emptyset$.}

	Formally, we say verifiability structure $M$ exhibits \emph{full verifiability} 
	if for each $s \in [0,1]$, $s \in M(s)$ and $s \notin M(s')$ for $s' \ne s$.
	That is: each type of sender has the possibility of proving her identity to the receiver.
	This setting includes the mandatory disclosure setting of \Cref{ex:mandatory-disclosure}.
	It also describes a key property of the canonical verifiability assumptions
	in \citet{Grossman1980}, \citet{Grossman1981} and \citet{Milgrom1981}.
	The outcome in this setting is discussed in \Cref{re:full-ver-opt}.
	\demo
\end{example}

The following environments exhibit instead \emph{partial} verifiability
in the sense that $M(s) \subseteq M(s')$ for some $s, s' \in [0,1]$,
meaning that type $s'$ can always imitate type $s$.

\begin{example}[Cheap talk]
	\label{ex:cheap-talk}
	At the other end of the spectrum from full verifiability,
	we have cheap talk situations,
	where the sender can claim anything regardless of what she discovered.
	This could be a consequence of the information acquisition technology:
	for example, it could be that learning
	only produces information which is `soft'
	or easy-to-falsify.
	Formally, we have that $M(s) = \bar{M}$ for each $s \in [0,1]$,
	for some given---perhaps very large---finite set $\bar{M}$.
	\Cref{sec:equilibria:pnbp} discusses outcomes
	in environments of this kind.
	\demo
\end{example}

\begin{example}[Partition]
	\label{ex:partition}
	Typically, what the sender can reveal is influenced by what she discovers
	(unlike in cheap talk) but she may be unable to unequivocally prove her findings.
	For example, third-party certification may only come in coarse grades
	(proving to the receiver that what the sender believes about the state
	lies in a certain set).

	To formalize this notion, let $\mathcal{P}$ denote a partition of $[0,1]$;
	a partitional verifiability structure $M_\mathcal{P}$
	assigns to each $P \in \mathcal{P}$ a message $m_P$.
	That is: if type $s$ lies in the partition element $P \subseteq [0,1]$,
	she must reveal the partition element she belongs to by sending message $m_P$.
	In applications, partitional structures are typically of the \emph{interval} form:
	each $P \in \mathcal{P}$ is an interval.
	In this case, the verifiability structure can be interpreted as a scoring protocol,
	with a higher score corresponding to a higher (non-overlapping) interval of types.
	\Cref{re:cs-examples} provides a general result
	for comparing such structures from the perspective of the sender.
	\demo 
\end{example}

\begin{example}[Dye verifiability]
	\label{ex:dye}
	Suppose that the sender can always make a verifiable claim
	about what she discovers, unless she learns nothing,
	in which case she has no evidence to present.
	Suppose also that, regardless of what she learns, she always
	has the option of suppressing her verifiable information
	by claiming that she learned nothing instead.
	This simple verifiability structure was introduced by
	\citet{Dye1985} and \citet{Jung1988} to show that
	unraveling can fail when verifiability is partial.
	In my model, this situation can be formalized by setting
	$M(s) = \{s, m_\emptyset \}$ for each $s \neq \overline{p}$,
	while $M(\overline{p}) = \{ m_\emptyset \}$.
	Message $m_\emptyset$ can therefore be interpreted
	as the `non-disclosure' message.%
	\footnote{To illustrate how unraveling fails
	with this verifiability structure
	in the case of a fixed (exogenous) signal,
	consider a strictly increasing $v$
	and a signal such that
	with probability $\alpha \in (0,1)$
	the sender learns the value of the state and
	with complementary probability $1 - \alpha$ learns nothing.
	This signal gives rise to three possible sender types:
	$0$, $\overline{p}$ and $1$.
	In equilibrium, sender type $0$
	will pool with type $\overline{p}$
	by sending message $m_\emptyset$,
	while type $1$ will separate by sending message $1$.
	This signal cannot, in fact, be an equilibrium in my model,
	as it is subject to a deviation
	which worsens the meaning of $m_\emptyset$.
	See \Cref{thm:pnbp}
	and the discussion in the continuation of
	\Cref{ex:dye} on \cpageref{ex:dye-cont} for more details.}

	This is a special case of a more general class
	of verifiability structures in which there is a set $S_v \subseteq [0,1]$
	of verifiable sender types such that $M(s) = \{s, m_\emptyset  \}$ for $s \in S_v$
	while $M(s) = \{ m_\emptyset \}$ for $s \notin S_v$.
	\demo
\end{example}

\begin{example}[Certifiable thresholds]
	\label{ex:thresholds}
	The sender might be able to make certifiable claims of the form
	`I discovered the state is high at least with probability $s$'.
	This report is credible provided no type $s' < s$ can also make such a claim.
	A special case of this is when the sender privately learns via a binary test
	and she can prove she obtained good news
	(e.g., she can prove her posterior is above the prior),
	but can always claim to have obtained no information.

	Verifiability structures of this kind
	can be described by \emph{certifiable thresholds}.
	Formally, they are defined by a sequence of sender types, $0 = s_0 < s_1 < \dots < s_N \le 1$,
	and a set of messages $\{m_0, \dots, m_N\}$
	such that $m_i \in M(s)$ iff $s \ge s_i$.
	In this environment, the sender can prove to the receiver that a threshold $s_i$
	(corresponding to a probability $s_i$ of $\omega = 1$) has been reached.
	She can always claim to have attained a lower threshold, but not a higher one.
	One such example was solved in detail in \Cref{sec:intro:example-overview}.
	\demo
\end{example}

\subsection{Strategies and solution concept}

In light of \Cref{re:unbias-belief-mes},
rather than writing $(\mathcal{S},\pi)$ for the sender's choice of signal,
it creates no ambiguity to suppress the set of realizations and only denote it by
the (finite-support) distribution $\pi$,
where it is understood that the implied set of signal realizations is
\begin{equation*}
	\mathcal{S} = \supp \pi_S = \{s \in [0,1] : \pi(0,s) + \pi(1,s) > 0\},
\end{equation*}
where $s \mapsto \pi_S(s) = \pi(0,s) + \pi(1,s)$
denotes the marginal distribution of the signal.

We can therefore denote a strategy for the sender by a pair $(\pi,\mu)$,
where $\pi$ is the signal chosen and
$\mu$ is her messaging mixed strategy:
$(\pi,s) \mapsto \mu(\pi,s)$ assigns to every $(\pi,s)$ an element of $\Delta(M(s))$.
Denote by $\beta: \mathcal{M} \to [0,1]$ the receiver's posterior belief about the state.
Given the triple $(\pi,\mu,\beta)$,
say that type $s \in [0,1]$ is on-path iff $s \in \supp \pi_S$
and that message $m \in \mathcal{M}$ is on-path
iff
\begin{equation*}
	m \in \supp \mu(\cdot|\pi,s) = \{ m \in M(s) : \mu(m|\pi,s) > 0 \}
\end{equation*}
for some $s \in [0,1]$ on-path.

The solution concept is in the spirit of perfect Bayesian equilibrium:
the sender acquires and transmits information optimally
given how the receiver forms beliefs
and the receiver's posterior belief about the state
is formed using Bayes' rule on the equilibrium path
and is consistent with the verifiability structure off-path.
As the sender's optimality conditions and
the receiver's on-path Bayesian updating are standard,
their formal statement is relegated to Appendix~\ref{app:eq-def-exist}.
Off-path, consistency of the receiver's belief
with the verifiability structure requires that
\begin{equation}
	\label{eq:conv-hull}
	\beta(m) \in \conv M^{-1}(m) \quad \text{for every} \quad m \in \mathcal{M},
\end{equation}
where $\conv M^{-1}(m)$ denotes the convex hull of $M^{-1}(m)$.

To see why this is the appropriate
equilibrium restriction on the receiver's off-path beliefs,
notice that following any message $m \in \mathcal{M}$,
the receiver's belief \emph{about the sender's type} must be supported on $M^{-1}(m)$.
Since the verifiability structure $M$
is such that beliefs determine messages
(\Cref{def:bdm}),
the receiver must also believe that the sender's
belief \emph{about the state} is an element of $M^{-1}(m)$,
so that her own belief about the state---$\beta(m)$---must reflect this
and lie in $\conv M^{-1}(m)$.

Finally, the following lemma
provides a continuity condition
on the verifiability structure
ensuring the existence of an equilibrium.

\begin{lemma}[Existence]
	\label{lemma:existence}
	If $s \mapsto \max_{m \in M(s)} \min M^{-1}(m)$ for $s \in [0,1]$
	is upper semi-continuous,
	an equilibrium exists.
\end{lemma}

The proof is in Appendix~\ref{app:eq-def-exist}.
The condition guarantees that the sender,
when facing a skeptical receiver assigning belief
$\min M^{-1}(m)$ to every message $m$,
has an optimal choice of signal.
A sufficient condition for this continuity requirement
is that the set of all possible messages $\mathcal{M}$ is finite.

Since this assumption is not necessary to state
the results in the next section, it is not maintained.
It will be explicitly invoked in \Cref{sec:cs},
where equilibrium existence is required
to state the comparative statics results.

\section{Unpersuasive and unraveling equilibria}
\label{sec:equilibria}

Recall that---as discussed in the previous section---we restrict attention to
non-decreasing and upper semi-continuous payoff functions $v$
and verifiability structures $M$ in which beliefs determine messages
(\Cref{def:bdm}) and such that \eqref{eq:inf-M} holds.
We will refer to parameters in this class as \emph{admissible}.

\begin{definition}
	\label{def:lc}
	For a given triple of admissible parameters $(v, \overline{p}, M)$,
	say that type $s \in [0,1]$ is \emph{lowest-consistent} with message $m$ iff $s = \min M^{-1}(m)$.
\end{definition}
That is, sender type $s \in [0,1]$ is lowest-consistent
with message $m$ if no lower type can send $m$.
Type $s$ can therefore separate from every lower
type---regardless of the receiver's belief---by sending $m$.%
\footnote{In the terminology of \citet{Seidmann1997},
sender type $s$ is lowest-consistent with some message $m$
if it is the lowest of the worst-case types for $M^{-1}(m)$.}
The maintained assumption that beliefs determine messages
further implies that the receiver's belief about $\omega = 1$
following message $m$ must be at least $s$
(both on- and off-path)
as captured by equilibrium condition~\eqref{eq:conv-hull}.

Note that \Cref{def:lc}
does not rely on the verifiability structure $M$
being such that beliefs determine messages.%
\footnote{It does, however, rely on the
signal realizations being in $[0,1]$
and on the sets of available messages
depending only on the signal realization
(and not directly on the signal).
As discussed in \Cref{re:unbias-belief-mes},
these restrictions entail no loss of generality
once $M$ is assumed to be
such that beliefs determine messages.}
If beliefs do not determine messages, however,
even if sender type $s \in [0,1]$ is lowest-consistent with message $m$,
the receiver could hold---off-path---any belief
\emph{about the state} following $m$.
This is because the sender faces no constraints
on the joint distribution of state and signal,
other than the requirement that the state marginal must equal the prior.
So, even if the receiver's belief
\emph{about the sender's type}
would still have to be supported
on types at least equal to $s$,
this would not translate
into any restriction on her belief
about the state.

The following terminology will also prove useful in stating the results.

\begin{definition}
	Given a sender strategy $(\pi,\mu)$ and receiver belief $\beta$,
	say that the receiver is \emph{maximally skeptical} following message $m$ iff $\beta(m) = \min M^{-1}(m)$.
	Say that the sender \emph{reveals all acquired information}
	iff for every $s \in \supp \pi_S$ and every $m \in \supp \mu(\cdot|\pi,s)$, $\beta(m) = s$.
\end{definition}

Equipped with these definitions,
we can turn to studying the properties of equilibria.

\begin{definition}[Unpersuasive equilibria]
	An equilibrium $(\pi,\mu,\beta)$ is \emph{unpersuasive} iff $v(\beta(m)) = v(\overline{p})$ for every message $m$ on-path.
\end{definition}

In unpersuasive equilibria,
any variation in the receiver's belief
is payoff-irrelevant for the sender:
she obtains the same payoff
she would achieve without
information acquisition and transmission,
with probability one.
The first result shows that,
in equilibria which are \emph{not} unpersuasive,
the sender's behavior and outcomes 
are tightly pinned down by the primitives of the model.
\begin{theorem}[Unraveling equilibria]
	\label{thm:unraveling}
	Suppose that $(\pi,\mu,\beta)$ is an equilibrium which is not unpersuasive.
	Then:
	\begin{enumerate}[label=(\roman*)]
		\item \label{thm:unraveling:lc} Every on-path sender type $s$ is lowest-consistent with every $m \in \supp \mu(\cdot|\pi,s)$.
		\item \label{thm:unraveling:skept} The receiver is maximally skeptical following every on-path message.
		\item \label{thm:unraveling:reveal} The sender reveals all acquired information.
	\end{enumerate}
	We call such equilibria \emph{unraveling}.
\end{theorem}
The proof is in Appendix~\ref{app:proof-main-thm}.
Notice that \Cref{thm:unraveling}~\ref{thm:unraveling:lc} is saying two things at once.
First, it states that only types that are lowest-consistent
with \emph{some} message can be on-path.
The verifiability structure $M$ therefore restricts the set of
possible outcomes of the sender's equilibrium information acquisition
to those that can then be credibly communicated to the receiver
(i.e., to types that are lowest-consistent with some message).
Second, it is saying that every on-path sender type will, indeed,
only use messages she is lowest-consistent with in equilibrium.

Note also that the theorem is not making the stronger claim
that the equilibrium set can be partitioned into
unpersuasive and unraveling equilibria,
as an equilibrium may be both.%
\footnote{To illustrate,
adjust the example in \Cref{sec:intro:example-overview}
by increasing the prior to $\overline{p} = 0.5$.
It can be verified that
there is an equilibrium in which
the sender acquires no information
and sends message $m_G$ on-path.
This equilibrium is both unpersuasive and unraveling.}

The logic behind the theorem is driven by the fact that,
if the receiver were not maximally skeptical
following an on-path message,
the sender could deviate to a signal
`worsening the meaning' of that message,
i.e., a signal such that the message
ends up being sent by a lower type.
This is profitable for the sender
as the set of on-path messages is unchanged,
but probability is shifted towards messages
that lead to higher payoffs.

This deviation arises whenever the receiver's equilibrium beliefs
vary in a way that affects the sender's payoff---i.e.,
in equilibria that are not unpersuasive.
It is precisely this variation in payoff
that gives the sender an incentive
to shift probability from `bad news' messages to `good news' ones.
This temptation evidently vanishes
when all on-path messages yield the same payoff.

Notice how this reasoning relies
on the assumption that beliefs determine messages,
which aligns the receiver's skepticism about the sender's type
with her skepticism about the state.
If type $s$ is lowest-consistent with message $m$,
the sender cannot \emph{arbitrarily} worsen its meaning,
as she must believe that $\omega = 1$ with probability
at least $s$ when sending $m$.
Consequently, the receiver's belief about the state following $m$
cannot fall below $s$,
as implied by equilibrium condition~\eqref{eq:conv-hull}.

Finally, this intuition also hinges on
the monotonicity and binary state assumptions,
as they ensure that worsening the meaning of a message
never corresponds to a decrease in the probability
that a `good news' message is sent.
Outside of the binary-state case
with monotone payoff,
this need not be the case.

\subsection{The sender's value in unraveling equilibria}

\Cref{thm:unraveling} also directly implies
that the sender obtains the same expected payoff across all unraveling equilibria.
This expected payoff can be expressed as the concave envelope of a skepticism-adjusted version of $v$.
Letting
\begin{equation}
	\label{eq:max-skept-value-posteriors}
	\underline{v}_M (s) \equiv v \left( \max_{m \in M(s)} \min M^{-1}(m) \right),
\end{equation}
for each $s \in [0,1]$, we have the following corollary of \Cref{thm:unraveling}.

\begin{corollary}[Unique value of persuasion]
	\label{cor:value}
	In every unraveling equilibrium
	the sender's expected payoff is $(\cav \underline{v}_M) (\overline{p})$.
\end{corollary}

That is, the sender's equilibrium expected payoff
is the concave envelope of the `value of interim sender beliefs',
when facing a receiver who is maximally skeptical following every message.
In the example from \Cref{sec:intro:example-overview},
these objects are depicted in \Cref{fig:3-act-example:cld}:
$\underline{v}_M$ is the solid black line,
while $\cav \underline{v}_M$ is the thick dotted line
and $(\cav \underline{v}_M) (\overline{p})$ is the black dot.

\begin{proof}[Proof of \Cref{cor:value}]
	Consider any unraveling equilibrium $(\pi,\mu,\beta)$
	and define $w_\beta(s) \equiv \max_{m \in M(s)} v(\beta(m))$, $s \in [0,1]$.
	The sender's equilibrium expected payoff is $(\cav w_\beta) (\overline{p})$,
	since---for a given $\beta$ and optimally chosen messages---the sender faces
	a canonical Bayesian persuasion problem
	with value of posteriors $w_\beta$.
	Since $w_\beta \ge \underline{v}_M$,
	we have that $(\cav w_\beta) (\overline{p}) \ge (\cav \underline{v}_M) (\overline{p})$.

	Observe next that $\underline{v}_M = w_\beta$ on $\supp \pi_S$.
	This holds because, for every $s \in \supp \pi_S$,
	\Cref{thm:unraveling}~\ref{thm:unraveling:lc} implies that $s$
	is lowest-consistent with some message,
	so that $\underline{v}_M(s) = v(s)$ by construction
	and \Cref{thm:unraveling}~\ref{thm:unraveling:reveal} implies that
	$w_\beta(s) = v(s)$.
	It follows that, in an auxiliary Bayesian persuasion problem
	in which the sender faces value of posteriors $\underline{v}_M$,
	the sender can obtain an expected payoff of
	$(\cav w_\beta) (\overline{p})$ by choosing signal $\pi$.
	Since $(\cav \underline{v}_M) (\overline{p})$
	is the highest expected payoff
	the sender can attain in this problem, it must be that
	$(\cav w_\beta) (\overline{p}) \le (\cav \underline{v}_M) (\overline{p})$.
\end{proof}

\subsection{Proving news better than the prior}
\label{sec:equilibria:pnbp}

We have so far established that any equilibrium
which is not unpersuasive must be unraveling,
meaning that it exhibits the properties described in \Cref{thm:unraveling}.
We now provide a sufficient condition on the primitives of the model
which ensures that all equilibria are unraveling.
\begin{definition}[PNBP]
\label{def:pnbp}
	For a given triple of admissible parameters $(v, \overline{p}, M)$, say that the sender can \emph{prove news better than the prior} iff $v(\min M^{-1}(m_h)) > v(\overline{p})$ for some $m_h \in \mathcal{M}$.
\end{definition}
Being able to prove news better than the prior means that
the sender can---by acquiring information in a way
that makes sending message $m_h$ possible---induce
a belief in the receiver
that she strictly prefers over the prior.
Importantly, this is possible irrespective of the receiver's belief
(i.e., even if she is maximally skeptical,
assigning belief $\min M^{-1}(m_h)$ to message $m_h$).

PNBP can hold only if,
in the terminology of \citet{Kamenica2011},
there is \emph{information the sender would share}.
That is, only if the sender is not already obtaining
her highest possible payoff at the prior.
Notice that in full-verifiability settings
(Examples~\ref{ex:mandatory-disclosure} and \ref{ex:full-verifiability})
this condition is also sufficient for PNBP.
In contrast, in cheap talk environments (\Cref{ex:cheap-talk})
the sender cannot prove news better than the prior
for any admissible $v$ and $\overline{p}$
(so even if there is information she would share).

\begin{theorem}[PNBP, unraveling and persuasion]
	\label{thm:pnbp}
	If the sender can prove news better than the prior,
	every equilibrium is unraveling.
	If instead the sender cannot prove news better than the prior,
	there exists a sender-preferred unpersuasive equilibrium
	in which she acquires no information.
\end{theorem}

The proof is in Appendix~\ref{sec:proof_of_theorem_pnbp}.
\Cref{thm:pnbp} implies that,
if the sender can prove news better than the prior,
her equilibrium expected payoff is unique,
because of \Cref{cor:value}.%
\footnote{PNBP does not ensure that
the sender's expected payoff in equilibrium is at least $v(\overline{p})$.
For instance, in the example
presented in \Cref{sec:cs:more-ver-hurt} 
verifiability structure $M''$ leads to 
a unique equilibrium which is unraveling
and such that the sender's expected payoff
is strictly below $v(\overline{p})$.}
In contrast, if the sender cannot prove news better than the prior,
any equilibrium which is not unpersuasive---i.e., which is unraveling,
in light of \Cref{thm:unraveling}---leads to
an expected sender payoff not higher than $v(\overline{p})$.%
\footnote{It is straightforward to construct examples
of such unraveling equilibria.
For example, modify the example in \Cref{sec:intro:example-overview}
by increasing the prior from $1/4$ to $9/20$,
so that PNBP no longer holds.
In this case, there is an unraveling equilibrium in which
the receiver is maximally skeptical following every message
and the sender obtains an expected payoff strictly below $v(\overline{p})$.}

When PNBP fails, therefore, the sender's bias in favor of state $1$
eliminates the possibility of valuable persuasion.
That is, in equilibrium the sender can never obtain an expected payoff
above the one she would achieve with no information acquisition (and transmission).
An immediate consequence is that in `cheap talk' environments
(\Cref{ex:cheap-talk}) covert learning
has no value for the sender.

Why is PNBP sufficient for unraveling?
In light of \Cref{thm:unraveling},
it is enough to argue that if PNBP holds and an equilibrium is unpersuasive,
then it must also be unraveling.
The intuition is as follows.
PNBP guarantees the existence of a message conveying `good news'---that is,
a message that yields the sender a payoff
strictly above $v(\overline{p})$,
even under maximal receiver skepticism.
In an unpersuasive equilibrium,
the sender must obtain $v(\overline{p})$ with probability one,
so she must be deterred from deviating to a strategy
in which she acquires information,
sends the `good news' message when she indeed discovers good news,
and otherwise falls back on an on-path message that delivers $v(\overline{p})$.

Such a deviation can be ruled out only if,
in equilibrium, the sender learns nothing,
and the prior itself is lowest-consistent
with some message $m$ the sender uses.
In this case, $m$ can no longer serve as a fallback
in the deviation just described.
But then, the unpersuasive equilibrium is also an unraveling one.

To conclude this section and illustrate its results,
it is informative to consider their implications
in a setting with Dye verifiability,
as introduced in \Cref{ex:dye}.

\begin{namedthm}[\Cref{ex:dye} {\normalfont (Dye verifiability, continued)}.]
\label{ex:dye-cont}
Assume that $v(p) = p$,
so that the sender obtains a payoff equal
to the expected value of the state.
This specialization of the model
brings it closest to the set-up studied in \citet{DeMarzo2019}.
Assume, to rule out trivial cases, that $\overline{p} \in (0,1)$.
My results directly imply the following:
(i) the sender's equilibrium expected payoff is equal to $v(\overline{p})$,
(ii) the receiver's belief following the non-disclosure message $m_\emptyset$
is equal to $0$ in all equilibria.

Point (i) holds because
an equilibrium exists%
\footnote{\Cref{lemma:existence}
cannot be invoked for equilibrium existence,
since the sufficient condition it relies on fails in this example.
It is immediate that an equilibrium exists, however.
For example, this is an equilibrium:
the sender perfectly learns the value of the state
(so $\pi_S(1) = \overline{p}$ and $\pi_S(0) = 1 - \overline{p}$)
and sends message $s$ when her type is $s \ne \overline{p}$;
the receiver's belief is $\beta(m_\emptyset) = 0$
and $\beta(s) = s$ for $s \in [0,1] \setminus \{ \overline{p} \}$.}
and all equilibria are unraveling
(because PNBP holds, so \Cref{thm:pnbp} applies)
so they all lead to the same expected payoff for the sender,
which is equal to $(\cav \underline{v}_M)(\overline{p}) = v(\overline{p})$,
by \Cref{cor:value}.

To see why (ii) holds, suppose otherwise that
the receiver's belief following the non-disclosure message
$m_\emptyset$ were equal to some $\beta(m_\emptyset) > 0$ in equilibrium,
implying that the sender's payoff
from sending message $m_\emptyset$ would strictly exceed $v(0)$.
The sender could then deviate to
acquiring full information,
revealing what she learns when $\omega = 1$
(by sending message~$1$---recall that $M(1) = \{1,m_\emptyset\}$)
and withholding her findings when $\omega = 0$
(by sending message~$m_\emptyset$---recall that $M(0) = \{0,m_\emptyset\}$).
This would lead to an expected payoff equal to
\begin{equation*}
	\overline{p} v(1) + (1-\overline{p}) v(\beta(m_\emptyset))
	> \overline{p} v(1) + (1-\overline{p}) v(0) = v(\overline{p}),
\end{equation*}
meaning that there cannot be an equilibrium
in which the sender's expected payoff
is equal to $v(\overline{p})$,
contradicting (i).

These two observations correspond to the
`minimum principle'
in \citet{DeMarzo2019} (their Proposition 1).
The most substantial difference
between this specialization of my model and theirs
is that they consider a setting
in which the sender is \emph{constrained}
in her information acquisition choice in the sense that
all available signals put positive probability
on the unverifiable signal realization.
Therefore, in their model
the non-disclosure posterior belief
may strictly exceed zero.
They show that this belief must be minimal, in equilibrium. 
In contrast, my sender is able to choose any signal,
so minimality of the non-disclosure posterior belief
means that it must be equal to zero. \demo
\end{namedthm}

\section{Verifiability comparative statics}
\label{sec:cs}

How does the sender's equilibrium payoff change
as the extent to which she can misreport her acquired information varies?
Which shifts towards `more verifiability' are desirable from the sender's perspective?

To address these questions,
I first provide a pre-order on the set of verifiability structures,
formalizing an intuitive notion of `more verifiability' in the context of the model.
I then show that if the sender can prove news better than the prior,
she obtains weakly better equilibrium outcomes under higher-ordered structures.
I illustrate that the proposed pre-order is the appropriate notion
for comparing verifiability structures in this context
by providing a converse:
unordered shifts may strictly hurt the sender,
even when she can prove news better than the prior.

I then show that the result is tight:
if the sender cannot prove news better than the prior,
meaning that unpersuasive equilibria exist,
increases in verifiability may---perhaps surprisingly---hurt her in equilibrium.
I conclude the section by characterizing verifiability structures
that are optimal for the sender.

\subsection{More verifiability}
The notion of `more verifiability' that I introduce
captures the idea that, for all sender
types, separation possibilities from lower types do not decrease.
More precisely, the requirement is that
if a sender type can separate from all lower types by means of a single message,
she must also be able to do so in the higher-ordered structure.

Given verifiability structure $M$,
let the \emph{lowest-consistent set} be
the set of types that are lowest-consistent with some message
(recall \Cref{def:lc})
and denote it by $L_M \subseteq [0,1]$.

\begin{definition}[Larger lowest-consistent set]
	Given verifiability structures $M''$ and $M'$,
	say that $M''$ has a \emph{larger lowest-consistent set than} $M'$,
	and denote it by $M'' \succeq^\mathrm{lc} M'$,
	iff $L_{M''} \supseteq L_{M'}$.
\end{definition}

Notice that $\succeq^\mathrm{lc}$ is a pre-order on the set of admissible verifiability structures.%
\footnote{$\succeq^\mathrm{lc}$ is not anti-symmetric.
For example, consider verifiability structure $M'$
with $M'(s) = \{m_\emptyset\}$ for each $s \in [0,1]$
and $M''$ with $M''(s) = \{m_\emptyset,\hat{m}\}$ for each $s \in [0,1]$.
Then $M'' \succeq^\mathrm{lc} M' \succeq^\mathrm{lc} M''$ but $M' \neq M''$.}
I discuss next two examples of $\succeq^\mathrm{lc}$-shifts.

\begin{example}[Partition refinement]
	\label{ex:partition-refine}
	As mentioned in \Cref{ex:partition},
	we can interpret interval-partition
	verifiability structures as scoring protocols:
	the sender can prove (perhaps via a third-party certifier)
	that her type lies in a certain interval,
	which corresponds to a score.
	It turns out that,
	if the sender chooses a certifier
	with a finer grid of scores
	(or if regulation imposes such a refinement),
	we have a  $\succeq^\mathrm{lc}$-increase
	in the verifiability structure.
	In this context, the comparative statics question can be cast as:
	when will the sender benefit, ex ante, from a finer certification technology?

	Formally, for some given left-closed
	partition $\mathcal{P}$ of $[0,1]$
	(that is: every $P \in \mathcal{P}$ is left-closed)
	let $M_\mathcal{P}$ denote a partitional verifiability structure as defined in \Cref{ex:partition}.
	If $\mathcal{P}''$ is a refinement of $\mathcal{P}'$
	then $M_{\mathcal{P}''} \succeq^\mathrm{lc} M_{\mathcal{P}'}$.
	For \emph{interval} partitions, also the converse holds:
	if $\mathcal{P}''$ is not a refinement of $\mathcal{P}'$
	then $M_{\mathcal{P}''} \not \succeq^\mathrm{lc} M_{\mathcal{P}'}$.
	\demo
\end{example}

\begin{example}[Adding a message]
	\label{ex:add-message}
	Another natural way in which a verifiability structure may change
	is if a new message is added (or removed).
	To illustrate this starkly,
	suppose we start in a situation in which all talk is cheap:
	any message can be sent by all sender types
	(for example because all evidence is easily falsified,
	or because no hard evidence can be presented for legal or technological reasons).

	The environment now changes:
	a non-falsifiable piece of evidence
	is discovered (and can be presented to the receiver)
	if and only if the sender learns that $\omega = 1$.
	This constitutes a `new message',
	which the sender can use
	only at the degenerate belief $s = 1$,
	and that she can therefore exploit
	to credibly convey her acquired information to the receiver.
	Adding a message in this sense corresponds to a $\succeq^\mathrm{lc}$-increase
	in the verifiability structure.

	To illustrate further, a `message'
	might also be added to an initial verifiability structure
	by allowing for `the right to remain silent'
	when it was previously not permitted.
	This corresponds, for example, to the sender adding a clause
	to her contract with a third-party certifier,
	allowing her to prevent him from publishing the results
	of the test after having observed them.

	More generally, fix some verifiability structure $M'$,
	some `new message' $m \notin \cup_{s \in [0,1]} M'(s)$
	and a set of types $S_m \subseteq [0,1]$.
	Construct $M''$ as follows:
	$M''(s) \equiv M'(s)$ if $s \notin S_m$ and $M''(s) \equiv M'(s) \cup \{m\}$ if $s \in S_m$.
	So $M''$ is the same as $M'$, except that types in $S_m$ have access to the `new message' $m$.
	It can be easily verified that $M'' \succeq^\mathrm{lc} M'$.
	\demo
\end{example}

\subsection{Comparative statics result}
Let $\underline{V}_M^*(v,\overline{p})$ and $\overline{V}_M^*(v,\overline{p})$
denote the lowest and highest equilibrium expected payoffs
attainable for the sender at triple $(v,\overline{p},M)$, respectively.
To ensure that these values are well-defined,
for the rest of this section
we will restrict attention to verifiability structures
satisfying the continuity requirement stated in the existence lemma
(\Cref{lemma:existence}).

Recall that \Cref{thm:pnbp} and \Cref{cor:value} together imply that
$\overline{V}_M^*(v,\overline{p})
= \underline{V}_M^*(v,\overline{p})
= (\cav \underline{v}_M) (\overline{p})$
if PNBP holds
and $\overline{V}_M^*(v,\overline{p}) = v(\overline{p})$ if it does not.

\begin{theorem}[More verifiability]
	\label{thm:cs}
	Consider verifiability structures $M'$ and $M''$. Then
	\begin{equation*}
		M'' \succeq^\mathrm{lc} M' \quad \Leftrightarrow \quad \overline{V}_{M''}^*(v,\overline{p}) \ge \overline{V}_{M'}^*(v,\overline{p})
	\end{equation*}
	for every admissible $v$ and $\overline{p}$ such that at $(v,\overline{p},M')$ the sender can prove news better than the prior.
\end{theorem}

The proof is in Appendix~\ref{sec:proof-thm-cs}.
The result can be directly applied to compare the verifiability structures
introduced in the context of the examples.

\begin{remark}
	\label{re:cs-examples}
	If the sender can prove news better than the prior,
	`adding a message' to any verifiability structure
	(as defined in \Cref{ex:add-message})
	is beneficial for her.
	That is also the case for refining partitional structures
	(as introduced in \Cref{ex:partition} and discussed further in \Cref{ex:partition-refine}).
	For partitional structures,
	the converse applies to interval partitions
	(see, again, \Cref{ex:partition} for a definition):
	if an interval partition is not a refinement of another,
	the sender will obtain a strictly better equilibrium outcome in the latter,
	for some admissible $v$ and $\overline{p}$.
\end{remark}

\Cref{thm:cs} also characterizes
which changes in the sender's commitment
to reveal acquired information
are valuable for her.
In this interpretation,
following signal realization $s$,
the sender can commit to using a message from the set $M(s)$.
Therefore, $\succeq^\mathrm{lc}$-increases in $M$
are exactly the changes
in the sender's commitment power
which are beneficial for her.

\Cref{thm:cs} follows straightforwardly from the following lemma
of independent interest
(also proved in Appendix~\ref{sec:proof-thm-cs}).

\begin{lemma}
	\label{lemma:CS}
	Consider verifiability structures $M'$ and $M''$. Then
	\begin{equation*}
		M'' \succeq^\mathrm{lc} M' \quad \Leftrightarrow \quad \underline{V}_{M''}^*(v,\overline{p}) \ge \underline{V}_{M'}^*(v,\overline{p})
	\end{equation*}
	for every admissible $v$ and $\overline{p}$ such that at $(v,\overline{p},M')$ an unraveling equilibrium exists.
\end{lemma}

Therefore, whenever an unraveling equilibrium exists
(regardless of whether PNBP holds or not)
a $\succeq^\mathrm{lc}$-increase in verifiability
cannot lower the sender-worst equilibrium payoff.
Conversely, changes in the verifiability structure that do not increase
(in the set-inclusion sense) which types are lowest-consistent
may do so.

\begin{remark}[More separation possibilities]
\label{re:more-sep}
An alternative (natural) notion of `more verifiability'
captures the idea that, for all sender types,
separation possibilities increase.
Formally, for a fixed verifiability structure $M$ and message $m \in \mathcal{M}$,
let $S_m^c \equiv [0,1] \setminus M^{-1}(m)$ denote the complement of $M^{-1}(m)$ in $[0,1]$.
That is, $S_m^c$ is the set of types that cannot send message $m$.
For a given $s \in [0,1]$, let
\begin{equation*}
 	K_M(s) \equiv \{ S \in 2^{[0,1]} : S = S_m^c \; \text{for some} \; m \in M(s)\}
\end{equation*}
denote the collection of type sets
from which $s$ can separate by means of a single message,
no matter what the receiver's belief is.
Given verifiability structures $M''$ and $M'$,
say that $M''$ exhibits \emph{more separation possibilities than} $M'$,
and denote it by $M'' \succeq M'$,
iff $K_{M''}(s) \supseteq K_{M'}(s)$ for all $s \in [0,1]$.

It can be easily verified that $M'' \succeq M'$ $\Rightarrow$ $M'' \succeq^\mathrm{lc} M'$
while $M'' \succeq^\mathrm{lc} M'$ $\not \Rightarrow$ $M'' \succeq M'$.
Shifts to a $\succeq$-higher structure are therefore \emph{sufficient}
to weakly improve the sender's equilibrium payoff
(provided she can prove news better than the prior)
but they are \emph{not necessary} (as highlighted in \Cref{thm:cs})
because of the monotonicity of $v$.
\end{remark}

\subsection{More verifiability may hurt the sender}
\label{sec:cs:more-ver-hurt}

At first glance,
moving to a $\succeq^\mathrm{lc}$-higher
(or, perhaps, even $\succeq$-higher, as introduced in \Cref{re:more-sep})
structure
might appear (via a standard replication argument)
to be \emph{always} beneficial for the sender,
regardless of whether she can prove news better than the prior
or whether an unraveling equilibrium exists.
The following example demonstrates that this is not the case:
the mere existence of a verifiable `good news' message can
taint the meaning of `claiming ignorance',
turning the inability to produce good news
into bad news---thereby hurting the sender.

Fix some finite set of messages $\bar{M}$
and some message $m_E \notin \bar{M}$.
Let verifiability structure $M'$ be such that $M'(s) = \bar{M}$ for each $s \in [0,1]$,
while $M''$ is such that $M''(s) = \bar{M}$ for $s < 0.9$
and $M''(s) = \bar{M} \cup \{m_E\}$ for $s \ge 0.9$.
$M'$ is therefore a cheap talk structure (as introduced in \Cref{ex:cheap-talk}).
$M''$ lets the sender use any of the cheap talk messages in $\bar{M}$,
but also allows her to prove excellent news
(namely: if she discovers that the state is more than 90\% likely to be $1$,
she can prove it by sending message $m_E$).
It is immediate that $M'' \succeq M'$
(and thus $M'' \succeq^\mathrm{lc} M'$).

Finally, let $v$ be defined as follows:
$v(p) = 0$ for $p < 0.4$,
$v(p) = 2$ for $p \in [0.4,0.8)$
and $v(p) = 3$ for $p \ge 0.8$.
Fix the prior to be $\overline{p} = 1/2$.

I now argue, applying directly the results from \Cref{sec:equilibria},
that the (unique) equilibrium payoff with structure $M''$
is strictly lower than the (unique) equilibrium payoff with structure $M'$.

Since under $M'$ the sender cannot prove news better than the prior,
\Cref{thm:pnbp} implies that the sender's highest expected payoff
in equilibrium is $v(\overline{p}) = 2$ (in this example, it is immediate that
this is the sender's unique equilibrium payoff),
as illustrated in \Cref{fig:more-ver-hurt:less-ver}
for the case in which there is a single cheap talk message: $\bar{M} = \{m_\emptyset\}$.
Under $M''$ however, the sender can prove news better than the prior,
so \Cref{thm:pnbp} and \Cref{cor:value} together imply
that the sender's equilibrium expected payoff
is $5/3 < 2$, as illustrated in \Cref{fig:more-ver-hurt:more-ver}.

\begin{figure}[ht]
	\centering
	\begin{subfigure}[h]{0.45\textwidth}
		\begin{tikzpicture}[scale=0.5]

  % === Axes ===
  \draw (0,0) -- (10,0);   % horizontal axis
  \draw[->] (0,0) -- (0,5.5); % vertical axis

  % === Tick marks & labels ===
  \draw (0,0.15)  -- (0,-0.15)   node[below] {\small $0$};
  \draw (10,0.15) -- (10,-0.15)  node[below] {\small $1$};
  \draw (4,0.15)  -- (4,-0.15)   node[below] {\small $0.4$};
  \draw (5,0.15)  -- (5,-0.15)   node[below] {\small $\overline{p}$};
  \draw (8,0.15)  -- (8,-0.15)   node[below] {\small $0.8$};
  \draw (0.15,20/6) -- (-0.15,20/6) node[left] {\small $2$};
  \draw (0.15,10/2) -- (-0.15,10/2) node[left] {\small $3$};

  % === Value function (light gray) ===
  \draw[very thick, gray] (0,0) -- (4,0);
  \draw[very thick, gray] (4,20/6) -- (8,20/6);
  \draw[very thick, gray] (8,10/2) -- (10,10/2);
  % Markers (light gray)
  \draw[fill, gray] (4,20/6) circle (0.075);
  \draw[fill, gray] (8,10/2) circle (0.075);
  \draw[draw=gray, fill=white] (4,0)    circle (0.075);
  \draw[draw=gray, fill=white] (8,20/6) circle (0.075);

  % === Skeptical value (black) ===
  \draw[thick] (0,20/6) -- (10,20/6);
  
  % === Equilibrium payoff ===
  \draw[draw=black, fill=black] (5,20/6) circle (0.15);

  % === Messages ===
  \draw (0,-1.5) -- (10,-1.5) node[right] {\small $m_\emptyset$};
  % Message ticks
  \draw (0,-1.4) -- (0,-1.6);
  \draw (10,-1.4) -- (10,-1.6);

  % Phantom message for spacing consistency in other figures
  \phantom{
  \draw (9,-2.2) -- (10,-2.2) node[right] {\small $m_E$};
  \draw (9,-2.1) -- (9,-2.3);
  \draw (10,-2.1) -- (10,-2.3);
  }

\end{tikzpicture}
		\caption{Verifiability structure $M'$}
		\label{fig:more-ver-hurt:less-ver}
	\end{subfigure}
	\qquad
	\begin{subfigure}[h]{0.45\textwidth}
		\begin{tikzpicture}[scale=0.5]

  % === Axes ===
  \draw (0,0) -- (10,0);   % horizontal axis
  \draw[->] (0,0) -- (0,5.5); % vertical axis

  % === Tick marks & labels ===
  \draw (0,0.15)  -- (0,-0.15)   node[below] {\small $0$};
  \draw (10,0.15) -- (10,-0.15)  node[below] {\small $1$};
  \draw (5,0.15)  -- (5,-0.15)   node[below] {\small $\overline{p}$};
  \draw (9,0.15)  -- (9,-0.15)   node[below] {\small $0.9$};
  \draw (0.15,10/3) -- (-0.15,10/3) node[left] {\small $2$};
  \draw (0.15,10/2) -- (-0.15,10/2) node[left] {\small $3$};

  % === Value function (light gray) ===
  \draw[very thick, gray] (0,0)       -- (4,0);      % v=0 on [0,0.4]
  \draw[very thick, gray] (4,10/3)    -- (8,10/3);   % v=2 on [0.4,0.8]
  \draw[very thick, gray] (8,10/2)    -- (10,10/2);  % v=3 on [0.8,1]
  % Markers (light gray)
  \draw[fill, gray] (4,10/3) circle (0.075);        % filled at (0.4, v=2)
  \draw[fill, gray] (8,10/2) circle (0.075);        % filled at (0.8, v=3)
  \draw[draw=gray, fill=white] (4,0)    circle (0.075); % hollow at (0.4, 0)
  \draw[draw=gray, fill=white] (8,10/3) circle (0.075); % hollow at (0.8, v=2)

  % === Skeptical value (black) ===
  \draw[thick] (0,0)        -- (9,0);       % skeptical = 0 for p<0.9
  \draw[thick] (9,10/2)     -- (10,10/2);   % skeptical = 3 for p>=0.9
  % Markers (black)
  \draw[draw=black, fill=white] (9,0)  circle (0.075); % open dot at (0.9, 0)
  \draw[fill] (9,10/2)               circle (0.075); % filled dot at (0.9, 3)

  % === Cav (dotted) ===
  \draw[ultra thick, dotted] (0,0) -- (9,10/2) -- (10,10/2);
  \draw[draw=black, fill=black] (5,25/9) circle (0.15); % Equilibrium payoff

  % === Messages ===
  \draw (0,-1.5) -- (10,-1.5) node[right] {\small $m_\emptyset$};
  \draw (9,-2.2) -- (10,-2.2) node[right] {\small $m_E$};
  % Message ticks
  \draw (0,-1.4) -- (0,-1.6);
  \draw (10,-1.4) -- (10,-1.6);
  \draw (9,-2.1) -- (9,-2.3);
  \draw (10,-2.1) -- (10,-2.3);

\end{tikzpicture}
		\caption{Verifiability structure $M''$}
		\label{fig:more-ver-hurt:more-ver}
	\end{subfigure}
	\caption{
	On the left:
	the sender value function
	(solid black)---overlaid over $v$
	(solid gray).
	On the right:
	the sender's skepticism-adjusted value function
	(solid black)---overlaid over $v$
	(solid gray)---and its concave envelope
	(thick, dotted).
	In each case, the sender's equilibrium expected payoff
	is represented by the large black dot.}
	\label{fig:more-ver-hurt}
\end{figure}

The reason why a replication argument fails here is that the sender
cannot commit to acquiring no information under $M''$,
thereby replicating her payoff under $M'$.
Informally, this is the case because `claiming ignorance' is credible under $M'$
(using any cheap talk message in $\bar{M}$),
but it is not under $M''$.
This is because under $M''$ the sender can prove good news (with $m_E$),
so she will seek it, only claiming ignorance
(using a message in $\bar{M}$) when she did not find said good news.
Increasing separation possibilities may therefore worsen the meaning of messages in equilibrium
(in this case: any message in $\bar{M}$ must lead to a payoff of $0$ under $M''$,
while it led to a payoff of $2$ under $M'$).
This can, in turn, lower the equilibrium expected payoff for the sender.

\subsection{Value of commitment and sender-optimal verifiability}

We showed that increasing verifiability,
in the sense of enlarging the set of sender types
that are lowest-consistent with some message,
is both necessary and sufficient to make the sender ex ante weakly better off,
provided that we are considering environments in which she can prove news better than the prior.

In light of this result,
it is straightforward to characterize the structures that are optimal for the sender,
in the sense that her equilibrium expected payoff is highest among all admissible ones.
Note first that the sender can obtain her full-commitment payoff
$(\cav v)(\overline{p})$
(that is: the payoff she would obtain under public information acquisition, as in \citet{Kamenica2011})
with the `mandatory disclosure' verifiability structure
(the one of \Cref{ex:mandatory-disclosure}: $M(s) = \{s\}$, $s \in [0,1]$).
This follows from \Cref{thm:pnbp} and \Cref{cor:value},
if she can prove news better than the prior,
and directly from \Cref{thm:pnbp}, if she cannot.%
\footnote{In more detail:
if with the mandatory disclosure verifiability structure
the sender cannot prove news better than the prior,
it must mean that $v(\overline{p}) = \max_{p \in [0,1]} v(p)$,
so that $(\cav v)(\overline{p}) = v(\overline{p})$.
This payoff is therefore attained in an unpersuasive equilibrium,
which exists by \Cref{thm:pnbp}.}

Therefore, characterizing sender-optimal structures is tantamount to characterizing structures
that allow her to obtain such full-commitment payoff:

\begin{proposition}[Sender-optimal verifiability]
	\label{prop:sender-optimal}
	For a given verifiability structure $M$, the following are equivalent:
	\begin{enumerate}[label=(\roman*)]
		\item In every equilibrium the sender obtains the full-commitment payoff for every admissible $v$ and $\overline{p}$.
		\item Every $s \in [0,1]$ is lowest-consistent with some message.
	\end{enumerate}
\end{proposition}

The characterization therefore highlights precisely in which sense
the full commitment assumption in \citet{Kamenica2011} can be relaxed
without affecting the sender's equilibrium payoff.
Provided that for every sender type there is a feasible message that no lower type can send,
her equilibrium expected payoff matches the full-commitment one.
If instead some type is not lowest-consistent with any message,
it is possible to find a set of admissible parameters
such that her equilibrium expected payoff is strictly lower than the full-commitment one.

\begin{remark}
	\label{re:full-ver-opt}
	Any full-verifiability environment as described in \Cref{ex:full-verifiability} is sender-optimal.
\end{remark}

\begin{proof}[Proof of \Cref{prop:sender-optimal}]
Let $M^{\textit{md}}$ denote the mandatory disclosure verifiability structure
($M^{\textit{md}}(s) = \{s\}$, $s \in [0,1]$).
Consider first (i) $\Rightarrow$ (ii).
Suppose that $M$ is such that some $s \in [0,1]$
is not lowest-consistent with any message.
Then $M \not \succeq^\mathrm{lc} M^{\textit{md}}$.
It follows from \Cref{thm:cs} that there exist admissible $v$ and $\overline{p}$ such that
the sender obtains a payoff strictly below the full-commitment one,
which is the unique outcome with $M^{\textit{md}}$.

Consider now (ii) $\Rightarrow$ (i).
If $v(\overline{p}) = v(1)$ the result is immediate,
since the sender obtains the maximal payoff of $v(\overline{p})$ in every equilibrium.
If $v(\overline{p}) < v(1)$,
the sender can prove news better than the prior with $M^{\textit{md}}$.
Consider any structure $M$ such that
each $s \in [0,1]$ is lowest-consistent with some message.
The sender can also prove news better than the prior with $M$,
so her expected equilibrium payoff is unique
(this follows from \Cref{thm:pnbp} and \Cref{cor:value}).
Since $M \succeq^\mathrm{lc} M^{\textit{md}}$,
she must obtain at least the same payoff with $M$
as she does with $M^{\textit{md}}$ by \Cref{thm:cs}.
This payoff cannot be strictly higher than the full-commitment one,
so they must coincide.
\end{proof}

\section{Receiver-optimal verifiability}
\label{sec:receiver-optimal}

Which verifiability structures are desirable for the receiver?
I show that a simple verifiability structure with two messages
is such that the sender provides full information in equilibrium.%
\footnote{I say that the sender \emph{provides full information} in equilibrium if,
conditional on $\omega = 1$, the receiver holds a posterior belief of $1$
(with probability one)
and, conditional on  $\omega = 0$,
the receiver holds a posterior belief of $0$ (with probability one).}
Provided that the receiver values information about the state,
this structure will therefore be optimal for her across all admissible structures.%
\footnote{In more detail: if the receiver obtains a payoff $u(p)$ at belief $p$,
obtaining full information about the state is optimal
(across all information structures)
if $u$ is convex.
That is the case, for example,
if $u(p)$ is the value of a decision problem under uncertainty
with expected-utility preferences:
there is an action set $A$ and a continuous function
$f : A \times \{0,1\} \to \mathbb{R}$
such that $u(p) = \max_{a \in A} p f(a,1) + (1-p) f(a,0)$
for each $p \in [0,1]$.}
I also provide a converse:
only structures in a specific binary class
are such that,
no matter what the (admissible) preferences and prior of the sender,
she will provide full information in every equilibrium.

\begin{proposition}[Full information provision]
	\label{prop:full-info}
	For a given verifiability structure $M$, the following are equivalent:
	\begin{enumerate}[label=(\roman*)]
		\item In every equilibrium the sender provides full information for every $(v,\overline{p})$ such that $v(\overline{p}) < v(1)$.
		\item Only types $0$ and $1$ are lowest-consistent with some message.
	\end{enumerate}	
\end{proposition}

$v(\overline{p}) < v(1)$ is a very mild requirement on what value the sender can potentially obtain from persuasion.
If otherwise $v(\overline{p}) = v(1)$---recall that $v$ is non-decreasing---the sender would be obtaining her highest possible payoff at the prior so, unsurprisingly, equilibria in which the sender does not provide full information persist.%
\footnote{If $v(\overline{p}) = v(1)$
and only types $0$ and $1$ are lowest-consistent with some message,
it can be shown that there exists
an equilibrium in which the sender provides full information.
However, for the sender such an equilibrium may be (substantially) worse than an unpersuasive equilibrium, which always exists and is sender-optimal, as she cannot prove news better than the prior (this follows from \Cref{thm:pnbp}).}

The class of verifiability structures
leading to full information provision takes an intuitive form.
It can be interpreted as the receiver committing
to only checking the sender's claim of `good news'
if the latter produces evidence that news is \emph{conclusively} good
(i.e., if the sender proves that the state is $1$).

In this situation, the only hope the sender has of swaying the receiver in her favor
is to learn with maximal probability that the state is $1$,
which is achieved by also learning that the state is $0$ with certainty.
Since type $1$ is lowest-consistent with some message
and $v(\overline{p}) < v(1)$,
the sender can prove news better than the prior.
Therefore, \Cref{thm:pnbp} implies that
all acquired information is revealed to the receiver,
so the latter also learns the value of the state.
For the receiver, committing to checking any other intermediate claim
provides the sender with more commitment power,
so may lead to less-than-full information provision.

\begin{proof}[Proof of \Cref{prop:full-info}]
	Consider first (ii) $\Rightarrow$ (i).
	Observe that the sufficient condition stated in \Cref{lemma:existence}
	is satisfied, so an equilibrium exists.
	If type $1$ is lowest-consistent with some message
	and $v(\overline{p}) < v(1)$
	then the sender can prove news better than the prior,
	so all equilibria must be unraveling (\Cref{thm:pnbp}).
	Therefore, in equilibrium,
	$\pi_S(1) = \overline{p}$ and $\pi_S(0) = 1 - \overline{p}$,
	since only $0$ and $1$ are lowest-consistent with some message.
	In any unraveling equilibrium the sender reveals all acquired information
	(\Cref{thm:unraveling}),
	hence the sender provides full information.

	Consider now (i) $\Rightarrow$ (ii).
	Suppose that some $s^* \notin \{0,1\}$
	is also lowest-consistent with some message.
	Let $v(p) = 1_{\{p \ge s^*\}}$ and $\overline{p} = s^*/2$.
	The following is an unraveling equilibrium
	in which the sender acquires
	and provides less-than-full information.

	The receiver is maximally skeptical, 
	assigning belief $\beta(m) \equiv \min M^{-1}(m)$
	to every $m \in \mathcal{M}$.
	The sender chooses a signal that splits the prior on $\{0,s^*\}$,
	sending any available message when $0$ realizes
	and sending a message that type $s^*$ is lowest-consistent with
	when $s^*$ realizes.
	Letting $w_\beta(s) \equiv \max_{m \in M(s)} v(\beta(m))$, $s \in [0,1]$,
	observe that the sender's strategy attains
	an expected payoff equal to
	$(\cav w_\beta)(\overline{p}) = (\cav v)(\overline{p}) = 1/2$,
	and is therefore optimal.
	The receiver's beliefs are updated using Bayes' rule on-path,
	and consistent with the verifiability structure by construction.
	This profile is therefore an unraveling equilibrium, as required.
\end{proof}

\begin{appendices}

\section{Equilibrium definition and existence}
\label{app:eq-def-exist}
\subsection{Equilibrium definition}
Fix a prior $\overline{p} \in [0,1]$
and sender payoff function $v: [0,1] \to \mathbb{R}$.
Let $\Pi$ denote the set of signals
(that is: $\Pi$ is the set of all
finite-support joint distributions
of state and signal realizations
with marginal distribution of the state
coinciding with the prior $\overline{p}$).
A strategy for the sender is a signal $\pi \in \Pi$
and a mixed messaging strategy
$(\pi,s) \mapsto \mu(\pi,s)$
with $\mu(\pi,s) \in \Delta (M(s))$
for each $\pi \in \Pi$ and $s \in [0,1]$
(recall that $M(s)$ denotes the finite set of messages
available to the sender at $s$).
A triple $(\pi,\mu,\beta)$ is an equilibrium if:

\begin{enumerate}[label=(\roman*)]
	\item \emph{Optimal information acquisition.}
	The sender chooses an optimal signal:
  	\begin{multline*}
		\int_{[0,1]} \int_{M(s)} v(\beta(m)) \mathrm{d}\mu(m|\pi,s) \mathrm{d}\pi_S(s) \ge \\ \int_{[0,1]} \int_{M(s)} v(\beta(m)) \mathrm{d}\mu(m|\pi',s) \mathrm{d}\pi_S'(s)
  	\end{multline*}
  	for every $\pi' \in \Pi$.

  	\item \emph{Sequentially rational communication.}
  	The sender chooses an optimal message: $v(\beta(m)) \ge v(\beta(m'))$
 	for every $m \in \supp \mu(\cdot|\pi',s)$
	and every $m' \in M(s)$,
	for every $\pi' \in \Pi$ and $s \in [0,1]$.
	
	\item \emph{Consistent receiver beliefs.}
	Beliefs are consistent with the verifiability structure:
	$\beta(m) \in \conv M^{-1}(m)$ for every $m \in \mathcal{M}$.
	If $m \in \mathcal{M}$ is on-path, it is further required that
	\begin{equation*}
		\beta(m) \int_{[0,1]} \mu(m|\pi,s) \mathrm{d}\pi_S(s) = \int_{[0,1]} s \mu(m|\pi,s) \mathrm{d}\pi_S(s).
	\end{equation*}
\end{enumerate}

\subsection{Equilibrium existence}
\subsubsection{Preliminaries}

Consider a fixed receiver belief function $\beta$. Define
\begin{equation}
	\label{eq:w-defn}
	w_{\beta}(s) \equiv \max_{m \in M(s)}v(\beta(m)).
\end{equation}
For any given bounded function $f: [0,1] \to \mathbb{R}$ and prior $\overline{p} \in [0,1]$ let
\begin{equation}
	\label{eq:s-_definition}
	s_{-}(\overline{p},f) \equiv \sup \{s \in [0,\overline{p}] : \cav f(s) = f(s) \}
\end{equation}
and
\begin{equation}
	\label{eq:s+_definition}
	s_{+}(\overline{p},f) \equiv \inf \{s \in [\overline{p},1] : \cav f(s) = f(s) \},
\end{equation}
where $\cav f$ denotes the smallest concave function that majorizes $f$.
The following two lemmas are immediate from \citet{Kamenica2011}.

\begin{lemma}
	\label{lemma:sender-best-reply-existence}
	If $w_\beta$ is upper semi-continuous, a sender-optimal signal exists and attains value $(\cav w_{\beta})(\overline{p})$.
	The signal $\pi$ with $\supp \pi_S = \{s_{-}(\overline{p},w_{\beta}),s_{+}(\overline{p},w_{\beta})\}$ is sender-optimal.
\end{lemma}

\begin{lemma}
	\label{lemma:affine-w}
	$\cav w_{\beta}$ is affine over $[s_{-}(\overline{p},w_{\beta}),s_{+}(\overline{p},w_{\beta})]$.
\end{lemma}

\subsubsection{Proof of Lemma~\ref{lemma:existence}}
The existence proof is divided into two parts:
\Cref{lemma:existence-no-better-prior}
and \Cref{lemma:existence-better-prior},
dealing separately with the cases in which
the sender can or cannot prove news better than the prior (\Cref{def:pnbp}).
\Cref{lemma:existence-no-better-prior} shows more than mere existence
as it constructs an unpersuasive equilibrium
in which the sender acquires no information,
which is used to prove the second part of \Cref{thm:pnbp}.
\Cref{lemma:existence-no-better-prior} does not rely
on the continuity assumption on $M$,
namely upper semi-continuity of
$s \mapsto \max_{m \in M(s)} \min M^{-1}(m)$, for $s \in [0,1]$,
stated in \Cref{lemma:existence}.
\Cref{lemma:existence-better-prior} does rely on it.

\begin{lemma}[Existence without PNBP]
	\label{lemma:existence-no-better-prior}
	If the sender cannot prove news better than the prior,
	an unpersuasive equilibrium
	in which the sender acquires no information exists.
\end{lemma}

\begin{proof}
	Fix some $m_0 \in M(\overline{p})$.
	The candidate equilibrium $(\pi,\mu,\beta)$ is defined as follows:
	\begin{enumerate}[label=(\roman*)]
		\item \label{list:no-info-eqm:signal}The signal $\pi$ is such that $\pi_S(\overline{p}) = 1$.
		
		\item $\mu(m|\pi',s) = 1$ for some (arbitrary) $m$ in $\argmax_{m \in M(s)} \min M^{-1}(m)$,
		for every $\pi' \in \Pi$ and $s \in [0,1]$ for which $m_0 \notin M(s)$;
		$\mu(m_0|\pi',s) = 1$ for every $\pi' \in \Pi$ and $s \in [0,1]$ for which $m_0 \in M(s)$.
		
		\item $\beta(m) = \min M^{-1}(m)$ if $m \neq m_0$ and $\beta(m_0) = \overline{p}$.
	\end{enumerate}
	Observe first that, by construction, the sender's messaging strategy is optimal whenever $m_0$ is not available
	(since $v$ is non-decreasing)
	and that the receiver's beliefs are consistent.
	It remains to show that for the sender it is optimal to acquire no information (\cref{list:no-info-eqm:signal} above) and send message $m_0$ at signal realizations where $m_0$ is available.
	Observe that,
	since the sender cannot
	prove news better than the prior
	and $v$ is non-decreasing,
	$v(\beta(m)) \le v(\beta(m_0))$ for all $m \in \mathcal{M}$.
	This in turn implies both that sending $m_0$ is optimal, when feasible, and that acquiring no information, as specified by \cref{list:no-info-eqm:signal},
	is an optimal information acquisition strategy,
	as message $m_0$ is then sent with probability one.
\end{proof}

\begin{lemma}[Existence with PNBP]
	\label{lemma:existence-better-prior}
	If the sender can prove news better than the prior
	and $s \mapsto \max_{m \in M(s)} \min M^{-1}(m)$,
	$s \in [0,1]$ is upper semi-continuous, an equilibrium exists.
\end{lemma}

\begin{proof}
	The candidate equilibrium $(\pi,\mu,\beta)$ is defined as follows.
	\begin{enumerate}[label=(\roman*)]
		\item The signal $\pi$ is such that $\supp \pi_S = \{s_-(\overline{p},w_\beta),s_+(\overline{p},w_\beta) \}$, where $w_\beta$, $s_-$ and $s_+$ are defined in equations \eqref{eq:w-defn}, \eqref{eq:s-_definition} and \eqref{eq:s+_definition}. 
		
		\item $\mu(m|\pi',s) = 1$ for some (arbitrary) $m$ in $\argmax_{m \in M(s)} \min M^{-1}(m)$ for every $\pi' \in \Pi$ and $s \in [0,1]$.

		\item $\beta(m) = \min M^{-1}(m)$ for every $m \in \mathcal{M}$.
	\end{enumerate}

	To streamline notation,
	write $s_-$ ($s_+$) to denote $s_-(\overline{p},w_\beta)$
	($s_+(\overline{p},w_\beta)$)
	as we will be holding $\overline{p}$,
	$\beta$ and $w_\beta$ fixed throughout the rest of the proof.
	Suppose first that $s_{-}$ ($s_{+}$)
	is lowest-consistent (\Cref{def:lc})
	with some message.
	Then, it follows from the definition of $\mu$,
	that $s_{-}$ ($s_{+}$) is lowest-consistent
	with every $m \in \supp \mu(m|\pi,s_{-})$
	($m \in \supp \mu(m|\pi,s_{+})$).

	By construction, the receiver's belief is therefore correct
	following on-path messages and consistent off-path.
	Given $\beta$, the sender's messaging strategy is optimal by construction,
	as it maximizes the receiver's posterior belief and $v$ is non-decreasing.
	Finally, notice that $w_\beta$ is upper semi-continuous.
	This follows from two observations.
	First,
	\begin{equation*}
		w_\beta(s)
		= \max_{m \in M(s)} v \big( \min M^{-1}(m) \big)
		= v \bigg( \max_{m \in M(s)} \min M^{-1}(m) \bigg),
	\end{equation*}
	where the first equality holds by construction
	(since $\beta = \min M^{-1}$)
	and the second equality follows from $v$ being non-decreasing.
	Second, $s \mapsto v \big( \max_{m \in M(s)} \min M^{-1}(m) \big)$
	is upper semi-continuous,
	since $v$ is upper semi-continuous and non-decreasing
	and $s \mapsto \max_{m \in M(s)} \min M^{-1}(m)$
	is upper semi-continuous by hypothesis.

	Therefore, since $w_\beta$ is upper semi-continuous,
	\Cref{lemma:sender-best-reply-existence} implies that $\pi$ is an optimal signal.
	It follows that, provided that each of $s_{-}$ and $s_{+}$
	is lowest-consistent with some message, the candidate is indeed an equilibrium.

	To complete the proof we must therefore show that $s_{-}$ and $s_{+}$ are indeed lowest-consistent with some message.
	The argument is divided into two separate cases.

	\begin{claim}
		If $s_{-} = s_{+} = \overline{p}$, then $s_{-}$ and $s_{+}$ are lowest-consistent with some message.
	\end{claim}
	\begin{proof}
		If $s_{-} = s_{+} = \overline{p}$ then it must be that $(\cav w_\beta) (\overline{p}) = w_\beta(\overline{p})$, by definition of $s_-$ and $s_+$.
		Suppose otherwise that $\overline{p}$ is not lowest-consistent with any message.
		Then there exists some $\tilde{s} < \overline{p}$ such that $w_\beta(\tilde{s}) \ge w_\beta(\overline{p})$,
		since any message available at $\overline{p}$ (including optimal ones) is also available at some lower type.
		Recall also that, since the sender can prove news better than the prior, there exists some message $m_h$ such that $w_\beta(\min M^{-1}(m_h)) > w_\beta(\overline{p})$ with $s_h \equiv \min M^{-1}(m_h) > \overline{p}$,
		since $v$ is non-decreasing.
		Summarizing, we have that $\tilde{s} < \overline{p} < s_h$ with $w_\beta(\tilde{s}) \ge w_\beta(\overline{p})$ and $w_\beta(s_h) > w_\beta(\overline{p})$.
		Since $\cav w_\beta$ is concave and majorizes $w_\beta$, it must be that $(\cav w_\beta) (\overline{p}) > w_\beta(\overline{p})$, contradicting that $(\cav w_\beta) (\overline{p}) = w_\beta(\overline{p})$.
	\end{proof}

	\begin{claim}
		If $s_{-} < \overline{p} < s_{+}$, then $s_{-}$ and $s_{+}$ are lowest-consistent with some message.
	\end{claim}
	\begin{proof}
		We make use of the following intermediate claim, proved after the main argument is complete.
		\begin{claim}
			\label{claim:w-increase}
			$w_\beta(s_{-}) < w_\beta(s_{+})$.
		\end{claim}
		Suppose first that $s_{-}$ is not lowest-consistent with any message.
		Then there exists a $\tilde{s} < s_{-}$ such that $w_\beta(\tilde{s}) \ge w_\beta(s_{-})$.
		Since $\cav w_\beta$ majorizes $w_\beta$ we have that $(\cav w_\beta)(\tilde{s}) \ge w_\beta(\tilde{s})$ and, by definition of $s_{-}$, $(\cav w_\beta)(s_{-}) = w_\beta(s_{-})$.
		It follows that $(\cav w_\beta)(\tilde{s}) \ge (\cav w_\beta)(s_{-})$.
		Observe next that \Cref{claim:w-increase} and the definition of $s_{-}$ and $s_{+}$ together imply that $(\cav w_\beta)(s_{+}) > (\cav w_\beta)(s_{-})$.
		Summarizing, we have that $\tilde{s} < s_{-} < s_{+}$ with $(\cav w_\beta)(\tilde{s}) \ge (\cav w_\beta)(s_{-})$ and $(\cav w_\beta)(s_{+}) > (\cav w_\beta)(s_{-})$, contradicting that $\cav w_\beta$ is concave.

		Suppose next that $s_{+}$ is not lowest-consistent with any message.
		Then there exists a $\tilde{s} < s_{+}$ such that $w_\beta(\tilde{s}) \ge w_\beta(s_{+})$.
		Observe next that the definition of $s_{-}$ and $s_{+}$ and \Cref{claim:w-increase} imply that $(\cav w_\beta)(s_{+}) = w_\beta(s_{+}) > w_\beta(s_{-}) = (\cav w_\beta)(s_{-})$.
		Since $w_\beta(\tilde{s}) \ge w_\beta(s_{+})$, we have that
		\begin{equation*}
			(\cav w_\beta)(\tilde{s}) \ge (\cav w_\beta)(s_{+}) > (\cav w_\beta)(s_{-}).
		\end{equation*}
		There are three possibilities.
		If $\tilde{s} < s_{-}$, then $\cav w_\beta$ is not concave, a contradiction.
		If $\tilde{s} = s_{-}$ then $(\cav w_\beta)(s_{-}) > (\cav w_\beta)(s_{-})$, also a contradiction.
		If $\tilde{s} \in (s_{-}, s_{+})$ then $\cav w_\beta$ is not affine over $[s_{-}, s_{+}]$, contradicting \Cref{lemma:affine-w}.
		It follows that both $s_{-}$ and $s_{+}$ are lowest-consistent with some message, completing the proof of the claim.
	\end{proof}

	\begin{proof}[Proof of \Cref{claim:w-increase}]
		Suppose otherwise that $w_\beta(s_{-}) \ge w_\beta(s_{+})$.
		Since the sender can prove news better than the prior,
		there exists a message $m_h$ such that $v(\min M^{-1}(m_h)) > v(\overline{p})$.
		Note that $s_h \equiv \min M^{-1}(m_h) > \overline{p}$,
		since $v$ is non-decreasing.
		It follows from the definition of $w_\beta$
		that $w_\beta(s_h) > v(\overline{p})$.

		Observe next that $w_\beta(s_{-}) \le v(\overline{p})$,
		since for any $m \in M(s_{-})$, $\beta(m) \le s_{-}$
		(recall that $\beta (m) = \min M^{-1}(m)$, by construction)
		and therefore $v(\beta(m)) \le v(s_{-}) \le v(\overline{p})$.
		This, in turn, implies that $w_\beta(s_h) > w_\beta(s_{-})$.
		Hence, since $w_\beta(s_{-}) \ge w_\beta(s_{+})$ by hypothesis,
		we have that $w_\beta(s_h) > w_\beta(s_{+})$ also holds.
		Summarizing, we have established that
		$w_\beta(s_h) > w_\beta(s_{-}) \ge w_\beta(s_{+})$.

		Recall that, by definition, $(\cav w_\beta)(s_{-}) = w_\beta(s_{-})$ and $(\cav w_\beta)(s_{+}) = w_\beta(s_{+})$.
		There are three possibilities to consider.
		If $s_h < s_{+}$, $\cav w_\beta$ is not affine over $(s_{-},s_{+})$, contradicting \Cref{lemma:affine-w}.
		If $s_h = s_{+}$ then $w_\beta(s_{+}) > w_\beta(s_{+})$, a contradiction.
		If $s_h > s_{+}$ then, since
		\begin{equation*}
			(\cav w_\beta)(s_h) > (\cav w_\beta)(s_{-}) \ge (\cav w_\beta)(s_{+}),
		\end{equation*}
		$\cav w_\beta$ is not concave, a contradiction.
		It follows that $w_\beta(s_{-}) < w_\beta(s_{+})$.
	\end{proof}

	This completes the proof of \Cref{lemma:existence-better-prior}.
\end{proof}

\section{Proof of Theorem~\ref{thm:unraveling}}
\label{app:proof-main-thm}
Let $(\pi,\mu,\beta)$ denote an equilibrium which is not unpersuasive
and suppose otherwise that there is an $s'' \in \supp \pi_S$
and a message $m'' \in \mu(\cdot|\pi,s'')$
such that, for some $s' < s''$, $m'' \in M(s')$.

Since $(\pi,\mu,\beta)$ is not unpersuasive,
some on-path message $m$ must lead to a payoff
$v(\beta(m)) \neq v(\overline{p})$ for the sender,
meaning that $\beta(m) \neq \overline{p}$.
It follows that the sender must be acquiring
some information in equilibrium,
so that $\supp \pi_S \ne \{\overline{p}\}$.
Let $\underline{s} \equiv \min \supp \pi_S$
and $\overline{s} \equiv \max \supp \pi_S$;
since $\supp \pi_S \ne \{\overline{p}\}$,
$\sum_{s \in \supp \pi_S} \pi_S(s) s = \overline{p}$ implies that
$\underline{s} < \overline{p} < \overline{s}$.
Define the sender's equilibrium interim value
from signal realization $s \in [0,1]$
as $w_\beta(s) \equiv \max_{m \in M(s)} v(\beta(m))$.

The argument relies on the following claim,
proved after the main argument is complete.

\begin{claim}
	\label{claim:affine-increasing-w}
	$w_\beta$ is affine and strictly increasing on $\supp \pi_S$.
\end{claim}

Suppose first that $s'' \le \overline{p}$.
$w_\beta$ being affine (from \Cref{claim:affine-increasing-w})
together with $\pi$ being a feasible signal,
meaning that $\sum_{s \in \supp \pi_S} \pi_S(s) s = \overline{p}$,
imply that the sender's equilibrium expected payoff can be written as
\begin{equation}
	\label{eq:on-path-payoff}
	\frac{\overline{s} - \overline{p}}{\overline{s} - s''} w_\beta(s'') + \frac{\overline{p} - s''}{\overline{s} - s''} w_\beta(\overline{s}).
\end{equation}
Consider now a deviating signal supported on $\{s',\overline{s}\}$ only.
By choosing this signal,
the sender obtains an expected payoff equal to
\begin{equation}
	\label{eq:deviating-payoff}
	\frac{\overline{s} - \overline{p}}{\overline{s} - s'} w_\beta(s') + \frac{\overline{p} - s'}{\overline{s} - s'} w_\beta(\overline{s}).
\end{equation}
To show that \eqref{eq:deviating-payoff} is strictly larger than \eqref{eq:on-path-payoff} it suffices to note that:
$w_\beta(s') \ge w_\beta(s'')$,
since $m''$ is optimal at $s''$ and $m'' \in M(s')$, by hypothesis;
$s' < s''$, by hypothesis;
and $w_\beta(\overline{s}) > w_\beta(s'')$, since $w_\beta$ is strictly increasing on $\supp \pi_S$, by \Cref{claim:affine-increasing-w}.
This contradicts that $\pi$ is an equilibrium signal.

Suppose next that $s'' > \overline{p}$.
There are three cases to consider and,
in each, a deviation analogous
to the one in the previous paragraph can be constructed:	
(a) if $s' > \overline{p}$,
a signal supported on $\{\underline{s}, s'\}$ is a profitable deviation;
(b) if $s' = \overline{p}$,
a signal supported on $\overline{p}$ only is a profitable deviation;
(c) $s' < \overline{p}$,
a signal supported on $\{s',s''\}$ is a profitable deviation.
This establishes \cref{thm:unraveling:lc} in the theorem.

To prove \cref{thm:unraveling:skept},
consider any on-path message $m$.
\Cref{thm:unraveling:lc} implies that message $m$
is sent by type $\min M^{-1}(m)$ only:
$m \in \supp \mu(\cdot|\pi,\min M^{-1}(m))$ and
$m \notin \supp \mu(\cdot|\pi,s)$
if $s \neq \min M^{-1}(m)$.
As the receiver's beliefs are
to satisfy Bayes' rule on the equilibrium path,
it must be that $\beta(m) = \min M^{-1}(m)$.

For \cref{thm:unraveling:reveal},
consider any on-path sender type $s \in \supp \pi_S$
and a message $m \in \supp \mu(\cdot|\pi,s)$.
\Cref{thm:unraveling:lc} implies that $s = \min M^{-1}(m)$
and \cref{thm:unraveling:skept} implies that $\beta(m) = \min M^{-1}(m)$.
Therefore, $\beta(m) = s$.
\qed

\begin{proof}[Proof of \Cref{claim:affine-increasing-w}]
	We first prove the `affine' part.
	Recall that $\underline{s} \equiv \min \supp \pi_S$,
	$\overline{s} \equiv \max \supp \pi_S$
	and that, since $\supp \pi_S \ne \{\overline{p}\}$,
	$\sum_{s \in \supp \pi_S} \pi_S(s) s = \overline{p}$ implies that
	$\underline{s} < \overline{p} < \overline{s}$.
	Let constants $a$ and $b$ be defined by
	$w_\beta(\underline{s}) = a + b \underline{s}$
	and $w_\beta(\overline{s}) = a + b \overline{s}$.

	Suppose that there is some $s \in \supp \pi_S$
	with $w_\beta(s) < a + b s$.
	By appropriately shifting probability mass
	from $s$ to $\{\underline{s},\overline{s}\}$,
	we can construct a signal $\pi'$
	such that $\pi_S'(s) < \pi_S(s)$ which leads
	to a strictly higher expected payoff for the sender.

	Suppose next that there is some $s \in \supp \pi_S$
	with $w_\beta(s) > a + b s$.
	By appropriately shifting probability mass
	from $\{\underline{s},\overline{s}\}$ to $s$
	we can construct a signal $\pi'$
	with $\pi_S'(s) > \pi_S(s)$ which leads
	to a strictly higher expected payoff for the sender.

	We now turn to the `strictly increasing' part.
	Suppose first that $w_\beta$ is constant on $\supp \pi_S$.
	It must then equal $v(\overline{p})$,
	since if it were strictly larger (smaller)
	the sender would be inducing in the receiver
	a posterior belief strictly above (below) $\overline{p}$
	with probability one (since $v$ is non-decreasing).
	But if $w_\beta$ is constant and equal to $v(\overline{p})$
	on $\supp \pi_S$, the equilibrium must be unpersuasive,
	which does not hold by hypothesis.

	Suppose next that there exist on-path signal realizations $s, s' \in \supp \pi_S$
	such that $s' > s$ and $w_\beta(s') < w_\beta(s)$.
	Consider any $m' \in \supp \mu(\cdot|\pi, s')$
	and any $m \in \supp \mu(\cdot|\pi, s)$.
	Since $v$ is non-decreasing,
	it must be that $\beta(m') < \beta(m)$.
	Since the receiver uses Bayes' rule following on-path messages,
	it must also be that
	$m' \in \supp \mu(\cdot|\pi, \tilde{s}')$
	for some $\tilde{s}' < s'$ in $\supp \pi_S$,
	or that $m \in \supp \mu(\cdot|\pi, \tilde{s})$
	for some $\tilde{s} > s$ in $\supp \pi_S$,
	or both.
	In all cases, there exists a pair of on-path sender types
	with the same value of $w_\beta$.
	Since we established that $w_\beta$
	is affine on $\supp \pi_S$,
	it must then be constant on $\supp \pi_S$,
	contradicting that
	$w_\beta(s') < w_\beta(s)$.
\end{proof}

\section{Proof of Theorem~\ref{thm:pnbp}}
\label{sec:proof_of_theorem_pnbp}

Start with the first part of the statement:
we want to show that if PNBP holds then all equilibria are unraveling.
Exploiting \Cref{thm:unraveling},
it is enough to show that if PNBP holds and an equilibrium is unpersuasive,
then it is also an unraveling equilibrium.
Let $(\pi,\mu,\beta)$ denote such an unpersuasive equilibrium.
Observe that, since PNBP holds,
there exists some $m_h$ such that $v(\min M^{-1}(m_h)) > v(\overline{p})$ and,
since $v$ is non-decreasing, we also have that $s_h \equiv \min M^{-1}(m_h) > \overline{p}$.

Suppose first that the sender acquires information,
so that $\supp \pi_S \ne \{\overline{p}\}$.
Let $\underline{s} \equiv \min \supp \pi_S$.
Since the sender acquires information, $\underline{s} < \overline{p}$.
Since the equilibrium is unpersuasive,
we have  that $v(\beta(m)) = v(\overline{p})$
for every $m \in  \supp \mu(\cdot|\pi,\underline{s})$.
Then the sender can profitably deviate by choosing signal $\pi'$
with $\supp \pi_S' = \{\underline{s},s_h\}$.
So, when PNBP holds,
there cannot be unpersuasive equilibria in which the sender acquires information.

Consider now the case in which
the sender acquires no information in equilibrium,
so that $\pi_S (\overline{p}) = 1$.
Suppose that there exists some $m'' \in M(\overline{p})$
such that $\mu(m''|\pi,\overline{p}) > 0$ for which we can find
some $s' < \overline{p}$ such that $m'' \in M(s')$.
That is: $\overline{p}$ is not lowest-consistent with message $m''$.

We will show that the sender can profitably deviate by choosing signal $\pi'$
with $\supp \pi_S' = \{s',s_h\}$.
Following realization $s'$, the sender obtains a payoff of at least $v(\overline{p})$ (since $m'' \in M(s')$, by construction) and following realization $s_h$ the sender obtains a payoff strictly above $v(\overline{p})$ (since $m_h \in M(s_h)$, $\beta(m_h) \ge s_h$ and therefore $v(\beta(m_h)) \ge v(s_h) > v(\overline{p})$).
Hence $\pi'$ is a strictly profitable deviation,
so if $(\pi,\mu,\beta)$ is an equilibrium in which the sender acquires no information,
it must be that $\overline{p}$ is lowest-consistent with every $m \in \supp \mu(\cdot|\pi,\overline{p})$.
It follows that $\beta(m) = \overline{p} = \min M^{-1}(m) $ for every $m \in \supp \mu(\cdot|\pi,\overline{p})$,
so that the receiver is maximally skeptical following every on-path message
and that the sender reveals all acquired information.
The equilibrium is therefore unraveling.

Now to the second part of the statement.
Such an unpersuasive equilibrium is constructed in \Cref{lemma:existence-no-better-prior},
Appendix~\ref{app:eq-def-exist}.
Note that \Cref{lemma:existence-no-better-prior} does not rely
on the continuity requirement stated in \Cref{lemma:existence},
which is not assumed for \Cref{thm:pnbp}.

To show that it is sender-preferred,
we will show that there is no equilibrium in which
the sender obtains an expected payoff strictly larger than $v(\overline{p})$.
Suppose that $(\pi,\mu,\beta)$ is such an equilibrium.
It cannot be unpersuasive by definition,
so \Cref{thm:unraveling} implies that it must be unraveling.
Therefore, \Cref{cor:value} implies that the sender's expected payoff
is equal to $(\cav \underline{v}_M) (\overline{p})$,
where $\underline{v}_M$ is defined in \eqref{eq:max-skept-value-posteriors}.
However, if $(\cav \underline{v}_M) (\overline{p}) > v(\overline{p})$
there must exist some $s \in [0,1]$ such that
$\underline{v}_M (s) > v(\overline{p})$,
implying in turn that there exists
an $m_h \in \mathcal{M}$ such that
$v(\min M^{-1}(m_h)) > v(\overline{p})$,
i.e., that the sender can prove news better than the prior,
a contradiction.
\qed

\section{Proof of Theorem~\ref{thm:cs}}
\label{sec:proof-thm-cs}

We prove first \Cref{lemma:CS}
and then turn to the proof of \Cref{thm:cs}.
Recall that in \Cref{sec:cs} we are maintaining
the condition stated in \Cref{lemma:existence},
which ensures equilibrium existence. 

\subsection{Proof of Lemma~\ref{lemma:CS}}
Start with the $\Rightarrow$ direction.
We use a replication argument
which exploits \Cref{thm:unraveling}.
Consider first the triple $(v,\overline{p},M')$
and any unraveling equilibrium signal $\pi'$,
which exists by hypothesis
and leads to payoff $\underline{V}_{M'}^*(v,\overline{p})$.
This is because, if PNBP holds, every equilibrium is unraveling
(by \Cref{thm:pnbp})
and the sender's expected payoff is unique
(by \Cref{cor:value}).
If PNBP does not hold,
the unraveling equilibrium expected payoff
does not exceed
the unpersuasive equilibrium payoff, from \Cref{thm:pnbp}.

We will show that,
in any equilibrium with parameters $(v,\overline{p},M'')$,
the sender's expected payoff from $\pi'$
is at least as large as $\underline{V}_{M'}^*(v,\overline{p})$.
This, in turn, means that the expected payoff
in any equilibrium with parameters $(v,\overline{p},M'')$
must also be at least as large as $\underline{V}_{M'}^*(v,\overline{p})$,
completing the argument.

Observe that in any unraveling equilibrium at parameters $(v,\overline{p},M')$,
\Cref{thm:unraveling} implies that for any equilibrium signal $\pi'$
at any $s \in \supp \pi_S'$ the sender obtains a continuation payoff of $v(s)$.
Since $s$ is lowest-consistent with some message under $M'$
and $L_{M''} \supseteq L_{M'}$,
$s$ is also lowest-consistent with some message under $M''$.
This in turn implies that,
for any equilibrium (unraveling or not) $(\pi'',\mu'',\beta'')$
when parameters are $(v,\overline{p},M'')$,
$\max_{m \in M''(s)} v(\beta''(m)) \ge v(s)$ for every $s \in \supp \pi_S'$.
Hence the sender's expected payoff from $\pi'$
with parameters $(v,\overline{p},M'')$
is at least as large as $\underline{V}_{M'}^*(v,\overline{p})$,
as required.

Now to the $\Leftarrow$ direction.
We will prove the contrapositive, i.e.,
that $M'' \not \succeq^\mathrm{lc} M'$ implies
$\underline{V}_{M''}^*(v,\overline{p}) < \underline{V}_{M'}^*(v,\overline{p})$
for some admissible $v$ and $\overline{p}$
such that at $(v,\overline{p},M')$ an unraveling equilibrium exists.
Observe first that $M'' \not \succeq^\mathrm{lc} M'$ implies that $L_{M'} \ne \{0\}$,
as $0 \in L_M$ for any $M$.
Since $L_{M'} \ne \{0\}$, there exists some $s^* \in L_{M'}$, $s^* > 0$ such that $s^* \notin L_{M''}$.

Take $v(p) = 1_{\{p \ge s^*\}}$ and $\overline{p} = s^*/2$.
Let $\underline{v}_{M'}$ and $\underline{v}_{M''}$ denote the value of posteriors with maximal receiver skepticism
as defined in equation~\eqref{eq:max-skept-value-posteriors} when the verifiability structures are $M'$ and $M''$, respectively.
With parameters $(v,\overline{p},M')$ PNBP holds,
so all equilibria are unraveling (\Cref{thm:pnbp})
and lead the sender to obtain an expected payoff
equal to her full-commitment payoff of $(\cav v)(\overline{p}) = 1/2$
(this follows from observing that
$(\cav v)(\overline{p}) = (\cav \underline{v}_{M'}) (\overline{p})$
and applying \Cref{cor:value}).

This expected payoff is not attainable in any equilibrium under $M''$.
If the sender cannot prove news better than the prior under $M''$,
this follows from \Cref{thm:pnbp},
as the largest expected payoff
for the sender in equilibrium is $v(\overline{p}) = 0 < 1/2$
in this case.
If the sender can prove news better than the prior under $M''$,
this follows because all equilibria must be unraveling (from \Cref{thm:pnbp}),
so $s^* \notin L_{M''}$ implies that $s^*$ cannot be on-path.
However, in an unraveling equilibrium,
$s^*$ must be on-path if the sender is to obtain
an expected payoff of $(\cav v)(\overline{p})$. \qed

\subsection{Proof of Theorem~\ref{thm:cs}}
For the $\Rightarrow$ direction,
observe that if at $(v,\overline{p},M')$
the sender can prove news better than the prior,
then an unraveling equilibrium exists (\Cref{thm:pnbp})
and the equilibrium expected payoff is unique (\Cref{cor:value}),
i.e., $\underline{V}_{M'}^*(v,\overline{p}) = \overline{V}_{M'}^*(v,\overline{p}) = (\cav \underline{v}_{M'}) (\overline{p})$.
\Cref{lemma:CS} therefore implies that
$\underline{V}_{M''}^*(v,\overline{p}) \ge \overline{V}_{M'}^*(v,\overline{p})$.
As $M'' \succeq^\mathrm{lc} M'$,
the sender can also prove news better than the prior at $(v,\overline{p},M'')$,
so that $\underline{V}_{M''}^*(v,\overline{p}) = \overline{V}_{M''}^*(v,\overline{p}) = (\cav \underline{v}_{M''}) (\overline{p})$
(from \Cref{thm:pnbp} and \Cref{cor:value}),
completing this part of the argument.

Now to the $\Leftarrow$ direction.
We will prove the contrapositive, i.e.,
that $M'' \not \succeq^\mathrm{lc} M'$ implies
$\overline{V}_{M''}^*(v,\overline{p}) < \overline{V}_{M'}^*(v,\overline{p})$
for some admissible $v$ and $\overline{p}$
such that at $(v,\overline{p},M')$ PNBP holds.
Observe that the construction used in the proof of \Cref{lemma:CS}
in fact considers a triple $(v,\overline{p},M')$
such that the sender can prove news better than the prior
and shows that the unique equilibrium expected payoff at $(v,\overline{p},M')$
is not attainable in any equilibrium under $M''$. \qed

\end{appendices}

\bibliography{bibl.bib}
\end{document}